\documentclass{article}
\usepackage[utf8]{inputenc}

\title{Searching for Regularity in Bounded Functions}
 \author{
   Siddharth Iyer\thanks{Research supported by NSF grant CCF-2131899}\\
   \texttt{siyer@cs.washington.edu}
   \and
   Michael Whitmeyer\thanks{Research supported by NSF grant CCF-2006359.}\\
  \texttt{mdwhit@cs.washington.edu}
}

\usepackage{amssymb,amsmath,amsthm, bbm, mathdots}
\usepackage{color, xcolor}
\usepackage{graphicx}
\usepackage{setspace}
\usepackage{xspace}
\usepackage{multirow}
\usepackage{array}
\usepackage{complexity}

\usepackage{dirtytalk}
\usepackage{tikz}
\usepackage{hyperref}
\usepackage{thmtools}
\usepackage{thm-restate}
\usepackage[capitalize,nameinlink]{cleveref}

\usepackage[shortlabels]{enumitem}

\hypersetup{colorlinks=true,urlcolor=blue,linkcolor=magenta,citecolor=[rgb]{.42,.56,.14},}

\theoremstyle{plain}
\newtheorem{theorem}{Theorem}[section]
\newtheorem{corollary}[theorem]{Corollary}
\newtheorem{proposition}[theorem]{Proposition}
\newtheorem{lemma}[theorem]{Lemma}
\newtheorem{claim}[theorem]{Claim}
\newtheorem{fact}[theorem]{Fact}

\newtheorem{direction}[theorem]{Direction}

\newtheorem{definition}[theorem]{Definition}

\theoremstyle{remark}
\newtheorem{remark}[theorem]{Remark}
\theoremstyle{plain}
\newclass{\DNF}{DNF}
\newclass{\DNFs}{DNFs}
\newclass{\ACzero}{AC^0}
\newclass{\TCzero}{TC^0}

\renewcommand{\R}{\mathbb{R}} 
\newcommand{\F}{\mathbb{F}} 

\renewcommand{\Pr}{\mathop{\bf Pr\/}}
\renewcommand{\E}{\mathop{\bf E\/}}

\newcommand{\abs}[1]{\left|#1\right|}

\newcommand{\inner}[1]{\langle{#1}\rangle}
\newcommand{\bigabs}[1]{\Bigl \lvert #1 \Bigr \rvert}
\newcommand{\bigbracket}[1]{\Bigl [ #1 \Bigr ]}
\newcommand{\bigparen}[1]{\left ( #1 \right )}
\newcommand{\ceil}[1]{\lceil #1 \rceil}

\newcommand{\norm}[1]{\| #1 \| }

\newcommand{\bigcurly}[1]{\left \{ #1 \right \}}

\newcommand{\bigvert}{\left | \right .}

\newfunc{\MAJ}{MAJ}
\newfunc{\MUX}{MUX}
\newfunc{\NAE}{NAE}
\newfunc{\OR}{OR} 
\newfunc{\AND}{AND}
\newfunc{\Tribes}{Tribes}
\newfunc{\LocalCorrect}{LocalCorrect}

\newfunc{\sgn}{sgn} 

\newfunc{\spar}{sparsity}
\newfunc{\rank}{rank}
\newfunc{\spn}{span}
\newfunc{\quasipoly}{quasipoly}
\newfunc{\Bias}{Bias}

\newfunc{\DT}{DTdepth} 
\newfunc{\DTs}{DTsize} 

\newcommand{\eps}{\varepsilon}

\renewcommand{\hat}{\widehat}
\renewcommand{\tilde}{\widetilde}

\newcommand{\pmone}{\{\pm 1\}}

\newcommand{\calA}{\mathcal{A}}
\newcommand{\calB}{\mathcal{B}}

\newcommand{\calS}{\mathcal{S}}
\newcommand{\calT}{\mathcal{T}}
\newcommand{\calU}{\mathcal{U}}
\newcommand{\calV}{\mathcal{V}}
\newcommand{\calW}{\mathcal{W}}

\newcommand{\bfz}{\mathbf{z}}

\newfunc{\Parity}{PARITY}

\newfunc{\Dict}{Dict}
\newfunc{\Corr}{Corr}
\newfunc{\avg}{avg}
\newfunc{\smooth}{smooth}
\newfunc{\dist}{dist}

\newclass{\ETH}{ETH}

\renewcommand{\epsilon}{\varepsilon}

\newcommand{\indicator}{\mathbbm{1}}

\newcommand{\paritykill}[1]{C_{\min}^{\oplus} [#1]}

\newcommand{\codim}{\mathrm{codim}}

\newcommand{\Restricted}[3]{{#1}_{#2 \Arrow #3}}

\newcommand{\Arrow}{%
\parbox{0.2cm}{\tikz{\draw[<-](0,0)--(0.2cm,0);}}
}

\newcommand{\tower}{\mathrm{twr}}
\newcommand{\regnum}{\mathsf{r}}

\newcommand{\uniform}{\mathsf{Unif}}

\newcommand{\transpose}{\mathsf{T}}

\usepackage[margin=1in]{geometry}

\usepackage[style=alphabetic, maxbibnames=5]{biblatex}
\begin{document}

\maketitle

\begin{abstract}

Given a function $f$ on $\F_2^n$, we study the following problem. What is the largest affine subspace $\calU$ such that when restricted to $\calU$, all the non-trivial Fourier coefficients of $f$ are very small? 

For the natural class of bounded Fourier degree $d$ functions $f:\F_2^n \to [-1,1]$, we show that there exists an affine subspace of dimension at least $ \tilde\Omega(n^{1/d!}k^{-2})$, wherein all of $f$'s nontrivial Fourier coefficients become smaller than $ 2^{-k}$. 
To complement this result, we show the existence of degree $d$ functions with coefficients larger than $2^{-d\log n}$ when restricted to any affine subspace of dimension larger than $\Omega(dn^{1/(d-1)})$. In addition, we give explicit examples of functions with analogous but weaker properties.

Along the way, we provide multiple characterizations of the Fourier coefficients of functions restricted to subspaces of $\F_2^n$ that may be useful in other contexts. Finally, we highlight applications and connections of our results to parity kill number and affine dispersers.
\end{abstract}

\section{Introduction}
\label{section:intro}

The search for structure within large objects is an old one that lies at the heart of Ramsey theory. 
For example, a famous corollary of Ramsey's theorem is that any graph on $n$ vertices must contain a clique or an independent set of size $\Omega(\log n)$. 
Another example is Roth's\footnote{The related Hales-Jewett theorem \cite{hales-jewett} is also a classic result in Ramsey theory.} theorem~\cite{roth} on $3$-term arithmetic progressions, which essentially says that \emph{every} subset of $\{1,\ldots,n\}$ of density $\delta > \Omega(1/\log \log n)$ must contain a $3$-term arithmetic progression.\footnote{See also the recent quantitative improvement due to Kelley and Meka \cite{kelley2023strong} which gives the same result for all subsets of density at least $\Omega(2^{-\log^{1/11}(n)})$}

Szemer\'edi's Regularity Lemma is also a well known example of this phenomenon. Roughly speaking, it states that any graph $G$ can be partitioned into $k := M(\delta)$ parts $V_1 ,\ldots,V_{k}$, wherein most pairs of parts $(V_i, V_j)$ are $\delta$-regular. In this setting, the $\delta$-regularity of $(V_i, V_j)$ roughly corresponds to saying that the bipartite graph induced across $V_i$ and $V_j$ appears as though its edges were sampled randomly. This powerful statement has found applications in both pure mathematics (e.g., Szemer\'edi's~\cite{szemeredi1975regular} generalization of Roth's result to $k$-term arithmetic progressions) and theoretical computer science (to test triangle-freeness in dense graphs~\cite{ruzsa-szemeredi-triangles, triangle-freeness-dense, shapira2006graph}).

Similar to the definition of regular partitions in Szemer\'edi's Regularity Lemma, one can also define a notion of regularity for functions. 
In particular, for functions $f : \F_2^n\to\R$, we follow Green \cite{green-regularity} and O'Donnell \cite{AOBFbook} and define a function to be \emph{$\delta$-regular} if all its nontrivial Fourier coefficients are at most $\delta$ in magnitude.\footnote{For a formal definition, see \Cref{def:delta-regular}, and for more background on Fourier analysis, see \Cref{section:prelim}.} 
This definition can be viewed as a pseudorandomness condition; 
in particular, a randomly chosen Boolean function $f:\F_2^n \to \pmone$ is $\delta$-regular with very high probability, even for $\delta = 2^{-\Omega(n)}$.\footnote{See for example \cite{AOBFbook}, Exercise 1.7 and Proposition 6.1.} 

The prior works surrounding graph regularity \cite{szemeredi1975regular, frieze-kannan-regularity, shapira2006graph} and function regularity \cite{green-regularity, hlms-regularity-lb} have been concerned with obtaining $\delta$-regular \emph{partitions}, which, roughly speaking, are partitions of the object at hand into (mostly) pseudorandom parts. 
Often, these results have quite poor dependencies on the parameter $\delta$ so as not to be practical for any reasonably small value of $\delta$ (see \Cref{thm:green-regularity} and \Cref{claim:ghz-partition} for detailed statements). 
Motivated by this, and by applications in theoretical computer science, we relax our requirement and look to find just \emph{one} $\delta$-regular part. 
Namely, we seek to understand the following quantity:
$$\regnum(f,\delta) := \min\{\codim(\calU): \calU \text{ is an affine subspace such that } f_\calU \text{ is }\delta\text{-regular}\},$$
where here and throughout this work $f_\calU: \calU \to \R$ denotes the restriction of $f$ to inputs coming from $\calU$.

Before stating our main results as well as prior work, we make a few remarks about the quantity $\regnum(f,\delta)$. 
In the special case when $\delta = 0$, the quantity $\regnum(f,0)$ has been previously studied in the literature, under the name of \emph{parity kill number} \cite{parity-kill}. 
This is the smallest number of parities that need to be fixed in order to make $f$ constant. 
The value $\regnum(f,0)$ is also a measure associated with affine dispersers, objects that have received significant attention in the study of pseudorandomness, see e.g. \cite{shaltiel-disperser, xinli-polylog-affine-dispersers, chattopadhyay-etal-affine-dispersers-logn, cohen-tal-f2polys, bensasson-kopparty-subspace-polys}.
An affine disperser of dimension $k$ is a coloring of $\F_2^n$ such that no affine subspace of dimension $k$ is monochromatic. 
If we view an affine disperser as a function $f:\F_2^n\to \{0,1,\ldots,C\}$, then its dimension is just $n-\regnum(f,0)+ 1$.

Now, we briefly discuss the bounds on $\regnum(f,\delta)$ most relevant to our work. For a general function $f:\F_2^n \to [-1,1]$, it is known that $\regnum(f,\delta) \leq 1/\delta$; this follows from a well-known density-increment argument, see \cite{meshulam} (for a short proof of this, see \Cref{claim:trivial-algorithm}).
One might ask if $\regnum(f,\delta)$ is small when we assume $f$ is structured, and a natural example of such functions is the class of functions with low Fourier degree. 
For general degree $d$ functions $f: \F_2^n\to [-1,1]$, the best bound on $\regnum(f,\delta)$ until this work was just the above mentioned bound of $1/\delta$. However, for the class of degree $d$ Boolean functions, we know that $\regnum(f,\delta) \leq \regnum(f,0) = O(d^3)$; this follows from the polynomial relationship between Fourier degree and decision tree depth, see \cite{Midrijanis-degree-certificate-cxty}, and \cite{shalev-thesis, buhrman-dewolf-boolean-fn-complexity} for surveys. 
We emphasize that this result relies crucially on Booleanity (and is independent of $\delta$), and one can ask if the more general class of degree $d$ functions bounded in the interval $[-1,1]$ also have small $\regnum(f,\delta)$ values. Our main result answers exactly this question, and provides an upper bound for $\regnum(f,\delta)$ in this setting.

\begin{restatable}{theorem}{main}
    \label{thm:bounded-deg-d-upper-bound-intro}
    For any $\delta\in (0,1)$ and any degree $d$ function $f:\F_2^n \to [-1,1]$, we have $\regnum(f,\delta) \leq n - \Omega\left( n^{1/d!}(\log(n/\delta))^{-2}\right).$
\end{restatable}

Note that the general bound $\regnum(f,\delta) \leq 1/\delta$ that we mentioned earlier, is only meaningful when $\delta > 1/n$, however, our theorem allows for $\delta$ to be much smaller. 
The regime of small $\delta$ is particularly interesting from the perspective of pseudorandomness. 
Indeed, in a qualitative sense, we see that by decreasing $\delta$, we are asking for affine subspaces where the restricted function looks increasingly like a random function. 
Using \Cref{thm:bounded-deg-d-upper-bound-intro} together with our connection between $\regnum(f,0)$ and the dimension of affine disperse, we obtain the following corollary which says that low degree polynomials cannot serve as good affine dispersers.

\begin{corollary}
\label{thm:rule-out-disperser}
If $f: \F_2^n\to \{0,\ldots,C\}$ has Fourier degree $d$, then  
$f$ cannot be an affine disperser of dimension $k$ for any $k \geq \Omega\left(n^{1/d!}(d + \log(nC))^{-2}\right)$.
\end{corollary}

\subparagraph*{Lower Bounds on $\regnum(f,\delta)$.} To complement \Cref{thm:bounded-deg-d-upper-bound-intro}, we present in \Cref{table:large-reg-nums} several examples of functions (bounded as well as Boolean) for which $\regnum(f,\delta)$ is large. For each row in the table, we exhibit a class of functions (whose degree and range is as specified), such that for any $\delta'\leq \delta$, \emph{no} affine subspace of dimension larger than $n - \regnum(f, \delta)$ is $\delta'$-regular. 

\begin{table}[h!]
\centering
\begin{tabular}{||c | c | c | c | c ||} 
 \hline
 $\delta \leq $ & $\regnum(f,\delta) \geq$ & $\deg(f)$ & $\textrm{range}(f)$ & Ref. \\ 
 [0.5ex] 
 \hline\hline \rule{0pt}{3ex}  
  $1/n$ & $n/2 - 1$ & 1 & $[-1,1]$ & \Cref{claim:degree-one-bounded-lb}\\
 $\binom{n}{d}^{-1}$ & $n - 2dn^{1/(d-1)}$ & $d$ & $[-1,1]$ & \Cref{claim:degree-d-bounded-lb} \\
 $\Theta(n^{-1/2})$ & $\Theta(\sqrt{n})$ & $n$ & $\pmone$ &  \Cref{lemma:majority-lower-bound}  \\
 $\frac{1}{2}\cdot n^{-d}$ (for $d \leq \frac{\log n}{\log \log n + 1}$) & $n - 2dn^{1/(d-1)}$ & $\Omega(n)$ & $\pmone$ & \Cref{cor:nonexplicit-boolean-simpler} \\
 $1/2^{2^k + 1}$ (for integer $k$) & $\Omega\left ((\log\frac{1}{\delta})^{\log_2(3)}\right )$ & $2^k$ & $\pmone$ & \Cref{claim:superlinear-boolean} \\
 [0.5ex] 
 \hline
\end{tabular}
\caption{Table of functions with large $\regnum(f,\delta)$ values.}
\label{table:large-reg-nums}
\end{table}
Observe that \Cref{claim:degree-d-bounded-lb} provides a somewhat of a converse to \Cref{thm:bounded-deg-d-upper-bound-intro}. However there is a noticeable gap between the two results, and we conjecture that \Cref{claim:degree-d-bounded-lb} is closer to being tight, and that \Cref{thm:bounded-deg-d-upper-bound-intro} could be improved. We also note that \Cref{claim:degree-d-bounded-lb} and \Cref{cor:nonexplicit-boolean-simpler} are not explicit -- it would be interesting to find more explicit examples.
\subsection{Related Work}
To the best of our knowledge, $\regnum(f, \delta)$ has not been explicitly studied before. However, it is closely related to well-studied notions of function regularity as well as the concepts of parity kill number and affine dispersers. 
In this section, we give a detailed description of both these connections.

\subparagraph*{Parity Kill Number and Affine Dispersers.}
As we have already mentioned, $\regnum(f,0)$ has been studied under the name of \textit{parity kill number}, denoted $\paritykill{f}$ (see \cite{parity-kill}). 
Parity kill number can be considered as a further generalization of the \textit{minimum certificate complexity} of $f$, denoted $C_{\min}[f]$, which is the minimum number of bits one must fix in order to make $f$ constant.
In particular, for any $\delta 
\geq 0$, we have $\regnum(f,\delta) \leq \regnum(f,0) \leq C_{\min}[f]$.
The minimum certificate complexity is one of several natural complexity measures that have been well studied for Boolean functions $f:\F_2^n\to \pmone$ (see \cite{buhrman-dewolf-boolean-fn-complexity, shalev-thesis} for surveys). 

As we have already alluded to, the quantity $\regnum(f,0)$ is also closely related to efficacy of $f: \F_2^n \to \{0,\ldots,C\}$ as an affine disperser. 
In the case of $C=1$, Cohen and Tal \cite{cohen-tal-f2polys} rule out $\F_2$-polynomials of degree $d$ as affine dispersers by showing that any such function satisfies $\regnum(f,0) \leq n - \Omega(d\cdot n^{1/(d-1)})$.
This result resembles our \Cref{thm:rule-out-disperser}; however, the two results are incomparable for two reasons. First, 
degree $d$ functions over $\F_2$ can have very large Fourier degree; 
moreover, the corresponding result of \cite{cohen-tal-f2polys} applies to functions whose range is $\F_2$, while ours applies to functions that take values in the set $\{0,\ldots,C\}$, which can have a much larger size. Furthermore, for $f: \F_2^n \to \{0,\ldots, C\}$, a standard argument (analogous to the one in \cite{Midrijanis-degree-certificate-cxty}) shows that $\regnum(f,0) \leq O(Cd^3)$, where $d$ here is the Fourier degree. However, this does not address the case where $C = \Omega(n)$, which is when \Cref{thm:rule-out-disperser} becomes useful.

\subparagraph*{Pseudorandom Partitions.}
As we have mentioned, much prior work on function regularity has been focused on finding pseudorandom partitions of $\F_2^n$.
To the best of our knowledge, the earliest result in this direction is due to Green \cite{green-regularity}; below, the notation $\tower(x)$ refers to an exponential tower of 2's $2^{2^{\iddots}}$ of height $x$.
\begin{proposition}[Theorem 2.1 in \cite{green-regularity}]
\label{thm:green-regularity}
For any $f: \F_2^n \to \{0,1\}$ and $\delta>0$, there exists a subspace $\calV$ of co-dimension $M(\delta) \leq \tower(\ceil{1/\delta^3})$ such that for all but a $\delta$-fraction of the affine subspaces $\calU = \alpha+\calV$, $f_\calU$ is $\delta$-regular.
\end{proposition}

In the same paper, Green showed that $M(\delta) \geq \tower(\Omega(\log (1/\delta)))$ was necessary. Subsequently, Hosseini et al. \cite{hlms-regularity-lb} exhibited a better counterexample showing co-dimension $M(\delta) \geq \tower(\ceil{1/16\delta})$ is required.

In the above upper and lower bound of \cite{green-regularity, hlms-regularity-lb}, the partition of $\F_2^n$ is of a specific form -- namely, it is every affine shift of a given subspace. Given this observation, one can ask if there is a partition of $\F_2^n$ into affine subspaces of smaller co-dimension so that in most parts $f$ is $\delta$-regular. 
As the next proposition, due to Girish et al. \cite{ghz-simpler} shows, this is indeed the case. 

\begin{proposition}[Proposition A.1 in \cite{ghz-simpler}]
\label{claim:ghz-partition}
For any $f: \F_2^n\to [0,1]$ and $\delta>0$, there exists a partition $\Pi$ of $\F_2^n$, where every $\pi \in \Pi$ is an affine subspace of co-dimension at most $\frac{1}{\delta^3}$ such that for all but a $\delta$-fraction of the parts, $f_\pi$ is $\delta$-regular.  
\end{proposition}
The proof of \cref{claim:ghz-partition} is based on a simple algorithm that greedily fixes the parities corresponding to the largest Fourier coefficients; it is included in \Cref{appendix:omitted} for completeness.

Although, both these results partition $\F_2^n$ into several affine subspaces where $f$ is $\delta$-regular, they are only meaningful when $\delta$ is relatively large. Indeed, \Cref{thm:green-regularity} is trivial when $\delta < (\log^*(n))^{-1/3}$, and \cref{claim:ghz-partition} when $\delta < n^{-1/3}$. 
As we mentioned earlier, if we relax our requirement to finding just one affine subspace, there is a simple upper bound on $\regnum(f,\delta)$ based on a density-increment argument, which goes back to the works of Roth \cite{roth} and Meshulam \cite{meshulam}.

\begin{proposition}[Folklore]
    \label{claim:trivial-algorithm}
    For any $f : \F_2^n \to [-1,1]$, we have $\regnum(f,\delta) \leq \frac{1}{\delta}$.
\end{proposition}
We provide a proof of \Cref{claim:trivial-algorithm} in \Cref{appendix:omitted} for completeness.

\subsection{Techniques}
\subparagraph*{Upper bound on $\regnum(f,\delta)$.} We give a brief proof sketch of~\Cref{thm:bounded-deg-d-upper-bound-intro}.
The proof proceeds by induction over the Fourier degree. The base case corresponds to degree one functions. 
Our intuition is derived from the following fact. If we have any $k$ real numbers $a_1,\ldots,a_k$ such that the sum of any subset of them has magnitude at most one, then by the pigeonhole principle, there is a non-empty subset $S \subseteq [k]$, and a signing of the numbers in $S$ so that the signed sum has magnitude at most $2^{-\Omega(k)}$. In the degree one case, we partition $\{1,\ldots,n\}$ into consecutive disjoint intervals of size $k = O(\log 1/\delta)$. 
We apply the above intuition to the $k$ Fourier coefficients in each interval, to obtain signed sums that have small magnitude. 
Then, by appropriately choosing an affine subspace, $\calU$ of dimension $\Omega(n/\log(1/\delta))$, we show that these signed sums are exactly the Fourier coefficients of the function restricted to $\calU$ (see \Cref{fact:Fourier-coeff-of f_U-as-a-sum-of-fourier-coeff-of-f-from-vperp} for a more general statement).
We give a more detailed description of how this works in \Cref{sec:proof-of-thm}.

At a high level, we reduce the problem for degree $d$ functions to degree $d-1$ by restricting to an affine subspace of dimension $\tilde{\Omega}\left(n^{1/d}\right)$, where the function is degree $d$ and all Fourier coefficients at the $d$-th level are extremely small $\ll \delta/n^d$. For a detailed statement, see \Cref{lemma:degree-d-small}.  When we use the inductive hypothesis for $d-1$, the last constraint ensures that the degree $d$ coefficients cannot increase the new coefficients by more than $O(\delta)$, even if they combine in the most constructive way possible. 

\Cref{lemma:degree-d-small} is also obtained by repeatedly applying the pigeonhole principle. However, the key issue now is that several Fourier coefficients could be affected when we apply a restriction, unlike the degree one case. 
To avoid this, we apply restrictions iteratively so that each one preserves the small Fourier coefficients from past iterations while still ensuring that \emph{several} new Fourier coefficients are also small. The cost of this procedure is that, in each step, we must apply the pigeonhole principle over larger and larger subsets of coordinates.

\subparagraph*{Lower Bounds.} 
Here, we give a very high level overview of our lower bounds on $\regnum(f,\delta)$. 
The basic idea is to consider functions $f$ with the property that their Fourier spectrum is concentrated on a small number of Fourier coefficients. 
It turns out (see \Cref{fact:Fourier-coeff-of f_U-as-a-sum-of-fourier-coeff-of-f-from-vperp}) that when we restrict to an affine subspace, say $\calU$, the Fourier coefficients of $f_{\calU}$ are simply signed sums of the Fourier coefficients of $f$. 
By our choice of $f$, if the restricted function was $\delta$-regular, then the large coefficients of $f$ involved in the signed sums somehow cancelled each other out. 
We show that by choosing the vectors corresponding to the large Fourier coefficients in $f$ appropriately, such a cancellation would imply that the co-dimension of $\calU$ must be large. For more detailed sketches of the entries in \Cref{table:large-reg-nums}, see \Cref{appendix:lb-sketches}.

\section{Preliminaries}
\label{section:prelim}
\subparagraph*{Notation.} 
$\indicator\{ \cdot\}$ denotes an indicator function that takes the value 1 if the clause is satisfied and 0 otherwise.
For a set $J\subseteq [n]$, we use $\spn(J)$ to denote the subspace spanned by the standard basis vectors corresponding to the elements in  $J$.
We refer to the $\text{L}_1$ norm of $\gamma \in \F_2^n$ by $\|\gamma\|_1$.
Given a subset $S \subseteq \F_2^n$, we denote $S^{=t} := S \cap \{u : \norm{u}_1 = t\}$.
Further, we define the degree of a function $f:\F_2^n \rightarrow \R$ to be $\max\{\|\gamma\|_1: \hat{f}(\gamma) \neq 0\}$.
We frequently interpret a linear transformation $M:\F_2^n\to\F_2^n$ as a matrix and refer to the linear map obtained by taking the transpose of the matrix as $M^{\transpose}$.
At several points, we consider the compositions of functions with linear maps. For a function $f$ and a map $M: \F_2^n \rightarrow \F_2^n$, we denote by $f\circ M$ the composition of the functions $f$ with $M$. In particular, $f\circ M(x)= f(M(x))$. 

\subparagraph*{Probability.} The following basic facts from probability theory are useful for us.
\begin{fact}[Hoeffding, \cite{Hoeffding}]
\label{fact:hoeffding}
Suppose $X_1,\ldots,X_n$ are such that $a \leq X_i \leq b$ for all $i$. Let $M = \frac{X_1 + \ldots X_n}{n}$. Then, 
$$\Pr\bigbracket{|M_n - \E M_n| \geq t} \leq 2\exp \bigparen{\frac{-2t^2n}{|b-a|}}.$$
\end{fact}

\begin{definition}[Statistical Distance]
\label{dfn:stat-dist}
Let $X$ and $Y$ be two random variables taking values in a set $\calS$. Then we define the statistical distance between $X$ and $Y$ as
$$|X-Y| := \max_{\calT \subseteq \calS}\bigabs{\Pr[X \in \calT] - \Pr[Y \in \calT]} = \frac{1}{2}\sum_{s \in \calS} \bigabs{\Pr[X = s] - \Pr[Y = s]}.$$
\end{definition}

\subparagraph*{Linear Algebra.}  



We recap two concepts from linear algebra, namely, orthogonal subspaces and direct sum, since they become useful for studying the Fourier spectrum of functions defined over subspaces of $\F_2^n$.
For a subspace $\calA$ of $\F_2^n$, we denote the \textbf{orthogonal subspace} of $\calA$ as $\calA^\perp = \{\gamma \in \F_2^n : \inner{\gamma, \gamma'} = 0, \: \forall \: \gamma' \in \calA\}$. We denote by $\dim(\calA)$, the dimension of $\calA$ and $\codim(\calA) = n - \dim(\calA)$.

We now define the notion of the direct sum of two subspaces. 

\begin{definition}[Independence, Direct Sum]
Two subspaces $\calA,\calB$ are \emph{independent} if $a+b\neq 0$ for any non-trivial choice of $a\in \calA$ and $b\in \calB$. 
In addition, if $\{a+b: a\in \calA \text{ and } b\in \calB\} = \F_2^n$, we say that $\F_2^n$ is a \emph{direct sum} of $\calA$ and $\calB$, written as $\calA\oplus\calB = \F_2^n$.\footnote{Such a subspace $\calB$ is sometimes called a \textbf{complement} of $\calA$. However, this term can be confused with the orthogonal subspace/complement, so we avoid using this terminology.}
\end{definition}

If $\calA\oplus\calB = \F_2^n$, then $\dim(\calA) + \dim(\calB) = n$. It is also well known that $\dim(\calA^{\perp}) + \dim(\calA) = n$. Note, however, that $\calA^{\perp}$ and $\calA$ need not be independent,\footnote{this might be unexpected at first for those used to working over the reals, but it is essentially because the inner product over $\F_2$ allows self-orthogonal vectors in $\F_2^n$.} and often in fact must not be.


\begin{fact}
\label{fact:cosets-unique}
Let $\calA,\calB$ be independent subspaces of $\F_2^n$. Then for all distinct $b , b' \in \calB $, the affine subspaces $b + \calA$ and $b' + \calA$ are mutually disjoint.
\end{fact}

\begin{proof}
If $b+a = b'+a'$, then a non-trivial sum of a vector from each $\calA$ and $\calB$ equals zero, contradicting the fact that $\calA\oplus \calB = \F_2^n$.
\end{proof}

\subparagraph*{Fourier Analysis.} 
For $f:\F_2^n \to \R$, we can write $f$ in the Fourier representation as 
$$f(x) = \sum_{\gamma \in \F_2^n}\hat{f}(\gamma)\chi_\gamma(x),$$
where $\chi_\gamma(x) := (-1)^{\langle \gamma,x\rangle}$ and $\hat{f}(\gamma) = \E_x[f(x)\chi_\gamma(x)]$. 
We say $f$ has degree $d$ if $\max_{\gamma: \hat{f}(\gamma) \neq 0} \| \gamma\|_1 = d$, and we refer to  the degree $d$ part of $f$ by $f^{=d} (x):= \sum_{\norm{\gamma}_1 = d} \hat{f}(\gamma)\chi_\gamma(x)$.
For more on this topic, see \cite{AOBFbook}, which uses notation consistent with ours. 

\subparagraph*{Restrictions.} 
We are ultimately concerned with understanding the Fourier coefficients of a function when it is restricted to some affine subspace of $\F_2^n$. 
In the special case where the coordinates in a set $J\subseteq [n]$ are fixed using the vector $b\in \F_2^J$, we denote the restriction of $f$ thus obtained as the function $\Restricted{f}{J}{b} : \spn(\overline{J}) \rightarrow \R$, which can be written as $\Restricted{f}{J}{b}(x) = f(x+b).$
Next, we recall the formula of the Fourier coefficients of the restricted function. 
Note that $\{\chi_\gamma(x) := (-1)^{\langle\gamma,x\rangle} : \gamma \in \spn(\overline{J})\}$ is a Fourier basis of the restricted function.

\begin{fact}[Fourier Coefficients of Restricted Functions (see \cite{AOBFbook}, Proposition 3.21)]
\label{fact:fourier-coeffs-restricted}
For every $\gamma\in \spn(\overline{J})$ and $b \in \spn(J)$,
$$\hat{\Restricted{f}{J}{b}}(\gamma) =
\sum_{\beta \in \spn(J)} \hat{f}(\beta + \gamma)\chi_{\beta}(b).
$$
\end{fact}

\subsection{Fourier Analysis on Subspaces}
We move to the general setting of restricting functions to arbitrary affine subspaces.\footnote{For an arbitrary subspace $\calV$, there is no canonical mapping between vectors and characters when $\calV\neq \F_2^n$, and we cannot simply define the vectors $\chi_\gamma$, for each $\gamma\in \calV$, as we did in the case of $\F_2^n$ to be the characters of $\calV$.} 
Let $\calU = \calV + \alpha$ be an affine subspace of $\F_2^n$. By the restriction of $f$ to $\calU$, we mean the function $f_\calU: \calV\rightarrow \R$ defined as $f_\calU(x) = f(x+\alpha).$ 

For the remainder of this section (and paper), let $\calW$ be such that $\calW \oplus \calV^\perp = \F_2^n$. 
For each element $\gamma\in \calW$, consider the function $\chi_\gamma: \calV\to \pmone$ as $\chi_\gamma(x) = (-1)^{\inner{\gamma,x}}$. 
It is easy to verify that $\{\chi_\gamma: \gamma\in \calW\}$ 
form an orthonormal basis of real-valued functions defined over $\calV$ under the inner product given by $\inner{p,q} = \E_{x\in \calV}[p(x)q(x)]$. We can therefore uniquely associate 
each vector $\gamma \in \calW$ with the function $\chi_\gamma$, and for $\calU = \alpha + \calV$, we can write 
\begin{align}
    \label{eq:fourier-decomp-affine-subspace}
    f_\calU(x) = \sum_{\gamma \in \calW}\hat{f_\calU}(\gamma)(-1)^{\inner{\gamma, x}}.
\end{align}

We now state the formal definition of $\delta$-regularity.

\begin{definition}[$\delta$-regularity]
\label{def:delta-regular}
Let $\calV$ be a subspace of $\F_2^n$ and $g: \calV \to \R$. For $\delta \geq 0$, we say $g$ is $\delta$-regular if
$\max_{\gamma \neq 0}\abs{\hat{g}(\gamma)} \leq \delta$.
\end{definition}

In this section, we present three separate formulas (\Cref{fact:coefficients-as-expectation}, \Cref{fact:Fourier-coeff-of f_U-as-a-sum-of-fourier-coeff-of-f-from-vperp} and \Cref{fact:affine-subspace-same-as-linear-transformation-with-restriction}) for the Fourier coefficients of $f_\calU$, each of which is useful in different contexts. 

First, using the above observations, we have the following simple formula for the Fourier coefficients of $f_{\alpha + \calV}$, which follows from the orthogonality of the $\chi_\gamma$ we have defined. 
\begin{fact}
\label{fact:coefficients-as-expectation}
Let $\calV,\calW$ be subspaces such that $\calW\oplus\calV^{\perp} = \F_2^n$ and $\calU = \alpha +\calV$.
For any $\gamma \in \calW$, we have that
$$\hat{f_{\calU}}(\gamma) = \E_{x \in \calV}[f(x+\alpha)\cdot(-1)^{\inner{\gamma, x}}] =  (-1)^{\inner{\gamma, \alpha}}\E_{x \in \calU}[f(x)\cdot(-1)^{\inner{\gamma, x}}].$$
\end{fact}

 \Cref{fact:coefficients-as-expectation} represents a simple and analogous formula for Fourier coefficients of functions restricted to affine subspaces. It also highlights that the magnitude of the Fourier coefficients of a restricted function are unaffected by the choice for shift $\alpha$ as long it corresponds to the same affine subspace.

Our next formula, which shows how the Fourier coefficients of $f_\calU$ can be written in terms of the Fourier coefficients of $f$, is an easy consequence of \Cref{fact:coefficients-as-expectation}.

\begin{proposition}
\label{fact:Fourier-coeff-of f_U-as-a-sum-of-fourier-coeff-of-f-from-vperp}
Let $\calV,\calW$ be subspaces such that $\calW\oplus\calV^{\perp} = \F_2^n$ and $\calU = \alpha +\calV$. For any $\gamma \in \calW$, we have
    $$\hat{f_{\calU}}(\gamma) = \sum_{\beta \in \gamma + \calV^\perp} \hat{f}(\beta)\cdot (-1)^{\inner{\beta, \alpha}} .$$
\end{proposition}

\begin{proof}
Using \Cref{fact:coefficients-as-expectation}, we can write
\begin{align*}
\hat{f_{\calU}}(\gamma)  = \E_{x\in \calV}[f(x+\alpha)\cdot (-1)^{\inner{\gamma,x}}]  
&= \E_{x\in \calV}\sum_{\beta}\hat{f}(\beta)(-1)^{\inner{\beta,x+\alpha}}(-1)^{\inner{\gamma,x}}\\
&= \sum_{\beta}\hat{f}(\beta)(-1)^{\inner{\beta,\alpha}}\E_{x\in \calV}[(-1)^{\inner{\beta+\gamma,x}}] \\
&= \sum_{\beta\in\gamma+\calV^{\perp}}\hat{f}(\beta)(-1)^{\inner{\beta,\alpha}},
\end{align*}
where the last equality follows by observing that $\E_{x\in\calV}\left[(-1)^{\inner{\gamma+\beta,x}}\right] = 1$ if $\beta\in \gamma+\calV^{\perp}$, and zero otherwise. 
\end{proof}

We note that \Cref{fact:Fourier-coeff-of f_U-as-a-sum-of-fourier-coeff-of-f-from-vperp} gives a formula analogous to \Cref{fact:fourier-coeffs-restricted} for restrictions to general affine subspaces.
This fact will be useful to construct functions and argue that they never become $\delta$-regular when restricted to any sufficiently large subspace.
%
Before we give our final formula, we highlight one particular choice of $\calW$ such that $\calW \oplus \calV^\perp = \F_2^n$.

\begin{definition}[$M$ mapping $\calV$ to $\spn(J)$]
\label{dfn:M}
Given a $k$-dimensional subspace $\calV$, let $B = \{\beta_1,\ldots,\beta_n\}$ be a basis for $\F_2^n$ such that $\calV = \spn(\{\beta_1,\ldots,\beta_k\}$. For any subset $J\subseteq [n]$ of size $k$, let $M: \F_2^n \to \F_2^n$ be an invertible linear map such that $\{M\beta_i:i\in[k]\} = \{e_j:j\in J\}$.
\end{definition}

\begin{proposition}[Choice of $\calW$]
\label{fact:choice-for-W}
Let $\calV$, $M$ and $J$ be defined as in \Cref{dfn:M}. The subspaces $\calW = \{M^{\transpose} \gamma: \gamma \in \spn(J)\}$ and $\calV^{\perp}$ are independent, and $\calW\oplus\calV^\perp = \F_2^n$.
\end{proposition}
\begin{proof}
We first show that $\calW$ and $\calV^{\perp}$ are independent.  Suppose that $M^\transpose \gamma + u = 0$, where $\gamma \in \spn(J)$ and $u\in \calV^{\perp}$. For any $v\in \calV$ such that $v \neq 0$, we have
$$0 = \inner{v,M^\transpose \gamma + u} = \inner{v,M^\transpose \gamma} = \inner{Mv,\gamma},$$
which is impossible unless $\gamma=0$ since this implies $Mv \in \spn(J)^{\perp} = \spn(\overline{J})$ and $Mv \neq 0$. This in turn implies that $u=0$ and therefore that $\calW$ and $\calV^{\perp}$ are independent. The claim follows by noting that $\dim(\calW\oplus\calV^{\perp}) = \dim(\calW) + \dim(\calV^{\perp}) = k+ n-k= n$.
\end{proof}
 
Finally, we show that the Fourier coefficients of a function restricted to an affine subspace are the same as the Fourier coefficients of the function $f \circ M$ under a suitable (normal) restriction and for a particular choice of $M$.

\begin{proposition}
\label{fact:affine-subspace-same-as-linear-transformation-with-restriction}
Let $\calV$, $M$ and $J$ be defined as in \Cref{dfn:M} and $\calU=\alpha+\calV$. For any $\gamma\in \spn(J)$, we have
$\left |\hat{f_{\calU}}(M^{\transpose}\gamma)\right | = \left|\hat{h_{\calU'}}(\gamma)\right|,$
where $h = f\circ M^{-1}$ and $\calU' = \{Mu: u\in\calU\} = M\alpha + \spn(J)$ is a standard restriction.
\end{proposition}
\begin{proof}
Repeatedly using \Cref{fact:coefficients-as-expectation}, we have that
\begin{align*}
\abs{\hat{f_\calU}(M^{\transpose}\gamma)} = \bigabs{\E_{x\in \calU}\left[f(x)(-1)^{\langle M^{\transpose}\gamma, x\rangle}\right]} 
&= \bigabs{\E_{x\in \calU}\left[f(x)(-1)^{\inner{\gamma, Mx}}\right]} \\
&= \bigabs{\E_{z\in\calU'}\left[f(M^{-1}z)(-1)^{\inner{\gamma, z}}\right]} = \abs{\hat{g_{\calU'}}(\gamma)}. \qedhere
\end{align*}
\end{proof}

\Cref{fact:affine-subspace-same-as-linear-transformation-with-restriction} implies the following important corollary.

\begin{corollary}
\label{cor:delta-reg-iff-exists-M-b}
    There exists an affine subspace $\calU$ of dimension $k$ such that $f_\calU$ is $\delta$-regular if and only if there exists an invertible linear map $M: \F_2^n \to \F_2^n$, a set $J\subseteq [n]$ of size $k$, and a fixing of coordinates outside $J$ given by $b\in  \F_2^{\overline{J}}$ such that the function $\Restricted{h}{\overline{J}}{b}$ is $\delta$-regular, where $h= f\circ M$.
\end{corollary}

We use \Cref{cor:delta-reg-iff-exists-M-b} crucially in the proof of \Cref{thm:bounded-deg-d-upper-bound-intro}, wherein we construct $M$ and $b$ such that $\Restricted{f\circ M}{\overline{[k]}}{b}$ has small Fourier coefficients.
In the proof of this theorem we must understand the Fourier coefficients of $f\circ M$ in terms of the Fourier coefficients of $f$. The following fact gives an identity relating the Fourier coefficients of the two functions. 

\begin{fact}[\cite{AOBFbook}, Exercise 3.1]

\label{fact:fourier-coeffs-linear-transformation}

Let $M$ be an invertible linear transformation, and consider the function $g = f \circ M^{-1}: \F_2^n \to \R$. Then we have that $\hat{g}(\gamma) = \hat{f}(M^{\transpose} \gamma).$

\end{fact}

\begin{proof}
We have that
\begin{align*}
    \hat{g}(\gamma) = \E_x[g(x) \chi_\gamma(x)] &= \E[f(Mx) \chi_\gamma(x)] \\
    &= \E_y[f(y) \chi_\gamma(M^{-1}y)] \\
    &= \E_y[f(y) \chi_{M^{-\transpose}\gamma}(y)] = \hat{f}(M^{-\transpose} \gamma),
\end{align*}
where we have used the fact that $\chi_\gamma(M^{-1}y) = (-1)^{\inner{\gamma,M^{-1} y}} = (-1)^{\inner{M^{-\transpose}\gamma, y}}$. 
\end{proof}

\section{Upper Bound on $\regnum(f,\delta)$}\label{sec:proof-of-thm}

 Now we prove our main theorem, restated here for convenience.
\main*

First, we gain some intuition from degree one functions.
\subparagraph*{Base Case/Toy Example.} Suppose $f$ is a Fourier degree one function. In this case our function has the form 
$$f(x) = \hat{f}(0) + \sum_i \hat{f}({e_i})(-1)^{x_i}.$$ 

For a parameter $t \geq 1$ and a subset $S\subseteq [t]$, consider the sum $g_S = \hat{f}(0) + \sum_{i \in S}\hat{f}(e_i)$. Note that $g_S = \E[f(x) | x_i = 0 \:\: \forall i \in S] \in [-1,1]$. 
The pigeonhole principle implies that for $t = \Omega(\log 1/\delta)$ there must exist two distinct sets $S,S'$ such that the difference $|g_S - g_{S'}| \leq \delta$. 
We can further write $g_S - g_{S'} = \sum_{i\in S\triangle S'}\hat{f}(e_i)(-1)^{|\{i\}\cap S'|}$.  

We now use the set $S\triangle S'$ and the signs to construct an affine subspace where at least one Fourier coefficient will have small magnitude. 
Assume without loss of generality that $1\in S\setminus S'$ and $S\triangle S' = [t']$ for some $t'\leq t$.
Consider restricting $f$ to the affine subspace $\calU$ defined by the linear equations $x_1 + x_i = b_i$ for each $i\in \{2,\ldots,t'\}$, where $b_i = |\{i\}\cap S'|$. 
We can reason about the Fourier spectrum of $f_\calU$ by plugging in $x_{i} = b_i + x_{1}$.
Under this restriction, we see that the Fourier coefficients of $e_{t'+1},\ldots,e_n$ stay the same, and the new Fourier coefficient of $e_1$ is exactly equal to 
$$\hat{f}(e_{1}) + \sum_{i = 2}^{t'} \hat{f}({e_i})(-1)^{b_i} = g_S - g_{S'},$$
which we observed has magnitude at most $\delta$. 
Repeatedly applying this argument roughly $n(\log(1/\delta))^{-1}$ times for the remaining standard basis vectors and fixing remaining coordinates arbitrarily, we obtain an affine subspace of dimension at least $\Omega\left (\frac{n}{\log(1/\delta)}\right)$.

\Cref{thm:bounded-deg-d-upper-bound-intro} is proved via induction using the following lemma.

\begin{lemma}
\label{lemma:degree-d-small}
For $\tau\in (0,1)$ and any degree $d$ function $f:\F_2^n \to [-1,1]$, there exists an invertible linear map $M: \F_2^n \rightarrow \F_2 ^n$, a set $J\subseteq[n]$ with size at least $\frac{d}{4e}\left(\frac{n}{\log 5/\tau}\right)^{1/d}$, and $b \in \spn (\overline{J})$ such that $h = f\circ M$ satisfies 
$$\abs{\hat{\Restricted{h}{\overline{J}}{b}}(\gamma)} \leq 
\begin{cases}
\tau \quad \text{ if } \|\gamma\|_1 = d,\\
0 \quad \text{ for all } \|\gamma\|_1 > d.
\end{cases}
$$
\end{lemma}

We now prove \Cref{thm:bounded-deg-d-upper-bound-intro} using \Cref{lemma:degree-d-small}.

\begin{proof}
The proof proceeds by induction over the degree. Our inductive hypothesis is that for any $\delta > 0$ and any degree $d$ function $f$, there exists an invertible linear map $M$, a set $I\subseteq [n]$, and $b\in \spn(\overline{I})$ such that the following two items hold:
\begin{enumerate}
    \item $\Restricted{h}{\bar I}{b}$ is $\delta$-regular, where $h = f\circ M$, and
    \item for 
    $C_d = \sum_{i=1}^d (i!)^{-1} $, we have 
    $$|I| \geq \frac{n^{1/d!}}{(8e)^{C_{d-1}}\bigparen{\log (n/\delta)}^{C_d}}.$$ 
\end{enumerate}

Note that $ C_d \leq e-1 < 2$ for all $d \geq1$.
The existence of the desired affine subspace is then given by \Cref{cor:delta-reg-iff-exists-M-b},
and its dimension is equal to $|I| \geq \Omega\bigparen{n^{1/d!} \left(\log(n/\delta)\right)^{-2}}$. 

The base case corresponds to the degree being one. Let us apply \Cref{lemma:degree-d-small} for degree one with $\tau = \delta$ and denote $g = f\circ M$, where $M$ is the linear map $M$ promised by the lemma. Additionally, we have a set $J$ of size at least $\frac{n}{4e\log 5/\delta} \geq \Omega\bigparen{\frac{n}{\log n/\delta }}$, and $b \in \spn(\overline{J})$ such that $$\abs{\hat{\Restricted{g}{\overline{J}}{b}}(\gamma)} \leq 
\begin{cases}
\tau \quad \text{ if } \|\gamma\|_1 = 1,\\
0 \quad \text{ for all } \|\gamma\|_1 > 1.
\end{cases}\implies \abs{\hat{\Restricted{g}{\overline{J}}{b}}(\gamma)} \leq \delta, \text{ for all } \gamma \neq 0.$$  

Assuming both items hold for some degree $d-1$, we show them for degree $d$. Applying \Cref{lemma:degree-d-small} with degree $d$ and $\tau = n^{-d}\delta/3$, we denote $p:= \Restricted{(f\circ M)}{\overline J}{b}$, 
where $M$, $J$ and $b$ are as promised by the lemma. 
Note that, by \Cref{lemma:degree-d-small}, $p$ has degree at most $d$, and for any $\gamma$ with $\|\gamma\|_1 = d$, we have, $\abs{\hat{p}(\gamma)} \leq \delta/(3n^d)$.
Consider the functions $p^{< d}$ and $p^{= d}$, which are the degree at most $d-1$ part of $p$ and the degree $d$ part of $p$, respectively.
We note that $\frac{p^{< d}}{(1+\delta/3)}$ is bounded in the interval $[-1,1]$ because for any $x$,
$$\abs{p^{< d}(x)} \leq \abs{p(x)} + \abs{p^{= d}(x)} \leq 1+ \sum_{\gamma: \|\gamma\|_1 =d} \abs{\hat{p}(\gamma)} \leq 1 + \frac{\delta}{3}.$$
Applying the inductive hypothesis\footnote{Technically, $\frac{p^{<d}}{1 + \delta/3} : \spn(J) \to [-1,1]$. However, we can abuse notation slightly and consider it as a function from $\F_2^{J}$ to $[-1,1]$ in order to apply the inductive hypothesis.} to $\frac{p^{<d}}{1 + \delta/3}$ for the choice of $\delta/3$, we get
a linear map $M'$, a set $I\subseteq J$, and $b'\in \spn(J \setminus I)$ such that $\Restricted{\left (\frac{q}{1 + \delta /3}\right )}{J\setminus I}{b'}$ is $\delta/3$-regular, where $q := p^{<d}\circ M'$. 
Therefore, for any $\gamma \neq 0$, we have $\abs{\hat{\Restricted{q}{J\setminus I}{b'}}(\gamma)} \leq \left(1+\frac{\delta}{3}\right)\frac{\delta}{3} < \frac{2\delta}{3}$.
Denoting $p' := p \circ M'$ and $r := p^{=d} \circ M'$, we have for any $\gamma \neq 0$ that
\begin{align*}
\bigabs{\hat{\Restricted{p'}{J\setminus I}{b'}}(\gamma)} \leq \bigabs{\hat{\Restricted{q}{J\setminus I}{b'}}(\gamma)} + \bigabs{\hat{\Restricted{r}{J\setminus I}{b'}}(\gamma)} < 2\delta/3 + \sum_{\beta:\|\beta\|_1 = d}\abs{\hat{g}(\beta)} \leq \delta.
\end{align*}
This shows that $\Restricted{p'}{J\setminus I}{b'}$ is $\delta$-regular. 
Moreover, if we extend $M'$ to act as the identity map on the coordinates in $\overline{J}$, we can write 
\begin{align*}
    \Restricted{p'}{J\setminus I}{b'}(x) = \Restricted{(p\circ M')}{J\setminus I}{b'}(x) &= p(M'(x + b')) \\
    &= \Restricted{(f\circ M)}{J}{b}(M'(x + b')) = f(MM'(x + b' + b)),
\end{align*}
which implies that item 1 of the inductive hypothesis is satisfied by applying the linear map $MM'$ and restricting to the set $I$ by fixing the coordinates outside according to $b+b'$.

We now show that the size of $I$ satisfies item 2 above. Note that~\Cref{lemma:degree-d-small} promises that $|J| \geq \frac{d}{4e}\left(\frac{n}{\log (15n^d/\delta)}\right)^{1/d}$. Moreover, we have 
$$\log (15n^d/\delta) \leq d\log n/\delta + \log 15 \leq 4d\log n/\delta,$$ 
where the last inequality follows for sufficiently large $n$. 
Therefore, $|J| \geq \frac{1}{8e}\left(\frac{n}{\log ( n/\delta)}\right)^{1/d}$. Moreover, we assume without loss of generality that $3|J| \leq n$ because, if not, we can arbitrarily fix coordinates in $J$ until it is, which does not affect the crucial property that all remaining degree $d$ Fourier coefficients have small magnitude. Using the bounds on $|J|$ and applying item 2 of the inductive hypothesis for degree $d-1$, we get 
\begin{align*}
    |I|\geq  \frac{|J|^{1/(d-1)!}}{(8e)^{C_{d-2}}\left(\log(3|J| /\delta)\right)^{C_{d-1}}} &\geq \frac{n^{1/d!}}{(8e)^{C_{d-1}}\log ( n/\delta)^{1/d!}\left(\log(3|J| /\delta)\right)^{C_{d-1}}}\\
    &\geq \frac{n^{1/d!}}{(8e)^{C_{d-1}}\left(\log(n /\delta)\right)^{C_{d}}}.
\end{align*}
This shows item 2 of the inductive hypothesis as desired.
\end{proof}

To prove \Cref{lemma:degree-d-small}, we need the following claim, which ultimately lets us bound Fourier coefficients in certain affine subspaces.

\begin{claim}[Pigeonhole Principle]
\label{claim:degree-d-subspace-construction-helper}
Let $f: \F_2^n \to [-1,1]$ be degree $d$.  For every $K \subseteq [n]$ of size $k$ such that $n-k \geq \binom{k}{d-1}\log(5/\tau)$, there exists $S \subseteq [n]\setminus K$ and $z \in \pmone^S$ such that \begin{enumerate}
    \item  $\forall \gamma \in \spn(K)$ with $\|\gamma\|_1 = d-1$, we have 
$\bigabs{ \sum_{j \in S}\hat{f}(\gamma + e_j)\cdot z_j} \leq \tau$, and
\item $1 < |S| \leq \binom{k}{d-1}\log(5/\tau)$.
\end{enumerate}
\end{claim}

\begin{proof}
 Consider any subset of $T \subseteq \overline{K}$ of size $\binom{k}{d-1}\log(5/\tau)$. For any $U \subseteq T$, consider the sum 
$$a_U(\gamma) := \hat{f}(\gamma)+ \sum_{j \in U} \hat{f}(\gamma + e_j).$$

We must have that $a_U(\gamma) \in [-1,1]$ since it is exactly equal to the Fourier coefficient corresponding to $\gamma$ if we restricted everything in $U$ to be one. This follows because $f$ is degree $d$. 

Now, divide the interval $[-1, 1]$ into $2/\tau$ intervals of length $\tau$. For a fixed $U \subseteq T$ of even size, consider putting the values of $a_U(\gamma)$ for all $\gamma \in \spn(K)^{=d-1}$ into a vector $v_U$ of length $\binom{k}{d-1}$. First, note that the number of even subsets of $T$ is at least $2^{\binom{k}{d-1}\log(5/\tau) - 1} > (2/\tau)^{\binom{k}{d-1}}$. Moreover, the number of possible interval vectors is at most $(2/\tau)^{\binom{k}{d-1}}$. Therefore, by the pigeonhole principle, there must be two distinct sets $U, U' \subseteq T$ such that $\norm{v_U - v_{U'}}_\infty \leq \tau$.

Thus, we have that 
$$\norm{v_U - v_{U'}}_\infty \leq \tau \Longleftrightarrow \sum_{i \in U \triangle U'}(-1)^{|\{i\} \cap U'|}\hat{f}(\gamma + e_i) \leq \tau \quad \forall \: \gamma \in \spn(K)^{=d-1}.$$
Since $U, U'$ have even size and are not equal, $U \triangle U'$ has even size as well, so we can set our $S = U\triangle U' \subseteq T$ and $z_i = (-1)^{|\{i\} \cap U'|}$, and the claim follows.
\end{proof}
We can now prove \Cref{lemma:degree-d-small}.
\begin{proof}[Proof of \Cref{lemma:degree-d-small}]
We build the map $M$, the set $J$, and the vector $b$ iteratively. Throughout the iterations, we seek to maintain a set $K$ of coordinates for which (under a suitable linear transformation $M$) every Fourier coefficient corresponding to a vector of weight $d$ in $\spn(K)$ has magnitude at most $\tau$. We build $K$ one coordinate at a time by repeatedly invoking \Cref{claim:degree-d-subspace-construction-helper} and arguing that the quantities guaranteed to be small by \Cref{claim:degree-d-subspace-construction-helper} are exactly the (new) Fourier coefficients. 
When we can no longer add more coordinates to $K$, we fix any remaining coordinates (outside of $K$ that are still alive), and we are left with a function, over only the coordinates in $K$, that has the desired property.

Note that we can start with $K$ being an arbitrary subset of size $d-1$ (w.l.o.g. let it be $[d-1]$) since any such subset has no Fourier coefficients of degree $d$. Therefore, we can assume without loss of generality that $\tau \geq 5\cdot 2^{-n/(4e)^d}$, since otherwise $\frac{d}{4e}\left(\frac{n}{\log(5/\tau)}\right)^{1/d} < d$ and the lemma becomes trivial. In each iteration, we maintain the following invariant for $M$, $J$ and $b$. In iteration $i$, there exists some $K\subseteq J$ of size $d+i-1$ such that the function $g = \Restricted{(f\circ M)}{\overline J}{b}$ satisfies
\begin{align*}
|\hat{g}(\gamma)|\leq 
\begin{cases}
\tau \quad \text{ if } \gamma \in \spn(K) \text{ and } \|\gamma\|_1 = d,\\
0 \quad \text{ for all } \|\gamma\|_1 > d.
\end{cases}
\end{align*}

Assume without loss of generality that $J = [j]$ for some $j\leq n$ and $K = [d+i-1] \subseteq J$. 
Since $g$ has degree $d$, we can apply \Cref{claim:degree-d-subspace-construction-helper} to $g$ and obtain a subset $S\subseteq  J \setminus K$ of size at most $\binom{d+i-1}{d-1}(\log (5/\tau)$ and a sign vector $z\in \{\pm1\}^S$ so that
\begin{align}~\label{eq:claim-application}
  \bigabs{ \sum_{j \in S}\hat{g}(\gamma + e_j)\cdot z_j} \leq \tau, \quad \text{ for all } \gamma\in \spn([d+i-1]) \text{ such that } \|\gamma\|_1  = d-1.  
\end{align}
We can also assume that $d+i\in S$ and $z_{d+i} = 1$. Now consider the invertible linear transformation $M_i:\F_2^n\to \F_2^n$ that maps $e_{d+i}$ to $\sum_{j\in S} e_j$ and behaves as the identity map on the remaining standard basis vectors. Further, denote $J_i := S \setminus \{d+i\}$ and let $b_i \in \spn(J_i)$, where $(b_i)_j := (1-z_j)/2$ for each $j\in J_i$. Intuitively, applying the linear transformation $M_i$ and then fixing the coordinates in $J_i$ to $b_i$ corresponds to restricting the affine subspace described by the equations $x_j + x_{d+i} = (1-z_j)/2$ for all $j \in J_i$.

After this iteration, we show that if we set $M' \leftarrow MM_i $, $J' \leftarrow J \setminus J_i$ and $b'\leftarrow b + b_i$, the invariant holds with $K' \leftarrow K \cup \{d+i\}$. For these choices, we have
\begin{align*}
\Restricted{(f\circ M')}{\overline {J'}}{b'}(x) =  f\circ M(M_i(x + b')) &= f\circ M(M_i(x + b_i + b)) \\
&= f\circ M(M_i(x + b_i) + b) \\
&= g\circ M_i(x+ b_i) = \Restricted{(g\circ M_i)}{J_i}{b_i}(x),
\end{align*}
and it therefore suffices to show that $\Restricted{(g\circ M_i)}{J_i}{b_i}$  -- denoted by $h$ henceforth, for shorthand -- is degree $d$ and $\abs{\hat{h}(\gamma)} \leq \tau$ for all $\gamma \in \spn([d+i])$ with $\|\gamma\|_1 = d$. We start by analyzing the Fourier coefficients of $h$, for which by \Cref{fact:fourier-coeffs-restricted} we have
\begin{align}~\label{eq:h-Fourier-formula}
    \hat{h}(\gamma) = \sum_{\beta \in \spn(J_i)} \hat{g\circ M_i}(\gamma +\beta)(-1)^{\langle\beta,b_i\rangle}.
\end{align}

Next, we observe the following relation between the Fourier coefficients of $g\circ M_i$ and those of $g$, which we use to simplify Equation~\eqref{eq:h-Fourier-formula}. Denoting $v := \sum_{j\in J_i} e_j $, we claim that, for any $\gamma$, 
\begin{align}~\label{eq:map}
    \hat{g\circ M_i}(\gamma) = \hat{g}(\gamma + e_{d+i}\langle\gamma, v\rangle).
\end{align}

Before proving Equation~\eqref{eq:map}, we use it to prove that $h$ has the desired properties.
Note that since $g$ is degree $d$, Equation~\eqref{eq:map} implies that if $\hat{g\circ M_i}(\gamma) \neq 0$, then $\|\gamma + e_{d+i}\langle\gamma,v\rangle\|_1 \leq d$, which in turn implies that $\|\gamma\|_1 \leq d+1$. 
This immediately tells us that $g\circ M_i$ has degree at most $d+1$; therefore, $h$ also has degree at most $d+1$ since the degree cannot increase under restrictions. Now, for any $\gamma$, Equation~\eqref{eq:h-Fourier-formula} reduces to
\begin{align}
    \hat{h}(\gamma) &= \sum_{\substack{\beta \in \spn(J_i),\\ \|\beta\|_1 \leq d+1 - \|\gamma\|_1}} \hat{g\circ M_i}(\gamma +\beta)(-1)^{\langle\beta,b_i\rangle} \nonumber \\
    &= \hat{g\circ M_i}(\gamma) + \sum_{\substack{\beta \in \spn(J_i),\\ 0 < \|\beta\|_1 \leq d+1 - \|\gamma\|_1}} \hat{g\circ M_i}(\gamma +\beta)(-1)^{\langle\beta,b_i\rangle}  \nonumber \\ 
    &=  \hat{g}(\gamma + e_{d+i}\langle\gamma, v\rangle) + \sum_{\substack{\beta \in \spn(J_i),\\ 0 < \|\beta\|_1 \leq d+1 - \|\gamma\|_1}} \hat{g}(\gamma +\beta + e_{d+i}\langle\gamma+\beta, v\rangle)(-1)^{\langle\beta,b_i\rangle}, \label{eq:h-Fourier-simple}
\end{align}
where, in the first equality, we used the fact that if $\|\beta\|_1 > d+1 - \|\gamma\|_1$, then $\|\beta + \gamma\|_1 > d+1$ and the corresponding Fourier coefficient in $g\circ M_i$ is just zero, and in the last equality, we used Equation~\eqref{eq:map}. Moreover, for any $\gamma\in \spn(J\setminus J_i)$, we have $\langle\gamma, v\rangle = 0$, which means that $\hat{g}(\gamma + e_{d+i}\langle\gamma, v\rangle) = \hat{g}(\gamma)$. 
We can now conclude that $h$ has degree at most $d$. Indeed, if $\|\gamma\|_1 \geq d+1$, then Equation~\eqref{eq:h-Fourier-simple} implies that $\hat{h}(\gamma) = \hat{g}(\gamma) = 0$ since $g$ has degree at most $d$. 

Next, we show that for any $\gamma\in \spn([d+i])$ with $\|\gamma\|_1 = d$, it must be that $|\hat{h}(\gamma)| \leq \tau$. Applying Equation~\eqref{eq:h-Fourier-simple} for such $\gamma$, we note that 
\begin{align*}
    \hat{h}(\gamma) &= \hat{g}(\gamma) + \sum_{j\in J_i} \hat{g}(\gamma + e_j + e_{d+i}\langle\gamma+e_j, v\rangle)(-1)^{\langle e_j,b_i\rangle}\\
    &= \hat{g}(\gamma) + \sum_{j\in J_i} \hat{g}(\gamma + e_j +  e_{d+i})z_j.
\end{align*}
We now consider two cases. First, when $\gamma_{d+i} = 0$, the above equation implies that $\hat{h}(\gamma) = \hat{g}(\gamma)$ since $\|\gamma + e_{d+i} + e_j\|_1 = d+2$ for every $j\in J_i$, and $g$ has degree at most $d$. Therefore, in this case, $\abs{\hat{h}(\gamma)} = \abs{\hat{g}(\gamma)} \leq \tau$ by the inductive hypothesis. Otherwise, $\gamma_{d+i} = 1$, and now using both Equation~\eqref{eq:claim-application} and the fact that $\gamma + e_{d+i} \in \spn(\{e_1,\ldots,e_{d+i-1}\})$, we conclude that 
$\abs{\hat{h}(\gamma)} =  \abs{\sum_{j\in S}\hat{g}((\gamma + e_{d+i}) + e_j)z_j} \leq \tau$.

It remains to show Equation~\eqref{eq:map}. We start by observing that $M_i = M_i^{-1}$, which can be verified by noting that $M_i^{-1}e_{d+i} = M_i^{-1}(e_{d+i} +v  +v) = e_{d+i} +v$ and $M_i^{-1}$ acts as the identity map on the remaining standard basis vectors. From \Cref{fact:fourier-coeffs-linear-transformation}, we know that $\hat{g\circ M_i}(\gamma) = \hat{g\circ M_i^{-1}}(\gamma) = \hat{g}(M_i^{\transpose}\gamma)$. Since the rows of $M_i^{\transpose}$ are the same as the columns of $M_i$, we have 
\begin{align*}
    (M_i^{\transpose}\gamma)_j = 
    \begin{cases}
    \langle v+e_{d+i},\gamma\rangle \quad &\text{ if } j = d+i,\\
    \gamma_j    \quad &\text{ otherwise.}
    \end{cases}
\end{align*}
Therefore, we can write $M_i^{\transpose}\gamma = \sum_{j\neq d+i}\gamma_je_j + e_{d+i}
\langle v+e_{d+i},\gamma\rangle = \gamma + e_{d+i}\langle v,\gamma\rangle$, as claimed.

We conclude the argument by calculating how many times we can repeat the above procedure. Note that, in the $i$-th iteration, we fixed at most $\binom{d+i-1}{d-1}\log 5/\tau - 1$ coordinates and we added exactly one coordinate to $K$. We can thus continue this process until iteration $t$ for the largest value of $t$ such that 
$$\log(5/\tau) \cdot\left(\sum_{i=1}^{t}\binom{d+i-1}{d-1}\right) \leq n - d + 1.$$
Simplifying the binomial sum, we get 
\begin{align*}
\sum_{i=1}^{t}\binom{d+i-1}{d-1} = \sum_{i=1}^{t}\binom{d+i-1}{i} &= \sum_{i=1}^{t}\binom{d+i-1}{i}  + \binom{d}{0} -  1 \\
&= \binom{d+t}{t} - 1 <  \left(\frac{e(d+t)}{ d}\right)^d,
\end{align*}
where the last equality follows by repeatedly using the identity $\binom{a}{i} + \binom{a}{i-1} = \binom{a+1}{i}$.
Thus, 
we can set $t = \frac{d}{e}\left(\frac{n-d+1}{\log 5/\tau}\right)^{1/d} - d$. Adding in the initial $d-1$ coordinates, at the end of the $t$ iterations, we can bound $|K|$ as,
\begin{align*}
    |K| &= \frac{d}{e}\left(\frac{n-d+1}{\log 5/\tau}\right)^{1/d} - d  + d-1 \\
    &\geq  \frac{d}{e}\bigparen{\frac{n}{\log 5/\tau}\bigparen{1-\frac{d-1}{n}}}^{1/d} - 1 \\
    &\geq \frac{d}{e}\bigparen{\frac{n}{\log 5/\tau}\cdot\frac{1}{d}}^{1/d} - 1 \\
    &\geq \frac{d}{2e}\bigparen{\frac{n}{\log 5/\tau}}^{1/d} - 1 \tag{$d^{1/d}\leq 2 \quad \forall d\geq 1$} \\
    &= \frac{d}{4e}\bigparen{\frac{n}{\log 5/\tau}}^{1/d} + \frac{d}{4e}\bigparen{\frac{n}{\log 5/\tau}}^{1/d} -1\\
    &\geq \frac{d}{4e}\bigparen{\frac{n}{\log 5/\tau}}^{1/d} + d - 1 \tag{since $\tau \geq 5\cdot 2^{-n/(4e)^d}$} \\
    &\geq \frac{d}{4e}\bigparen{\frac{n}{\log 5/\tau}}^{1/d} .\tag{$d \geq 1$}
\end{align*}

At the end of $t$ iterations, 
we can fix any coordinates outside the set $K$ arbitrarily to ensure that the only non-zero Fourier coefficients with $\text{L}_1$ norm $d$ in the resulting function must correspond to vectors in $\spn(K)$, which do not change under the restriction. 
\end{proof}

\section{Lower Bounds on $\regnum(f, \delta)$ }
In this section, we prove lower bounds on $\regnum(f,\delta)$. We start with lower bounds for functions $f$ that are bounded in the interval $[-1,1]$; in the subsequent section, we give lower bounds for Boolean functions.
\subsection{Bounded Functions}
We begin with a simple bound on the number of standard basis vectors in low-dimensional affine subspaces, which is crucial in the analysis of the lower bounds.
\begin{claim}
\label{claim:std-basis-vecs-in-small-subspace}
For any subspace $\calV\subseteq \F_2^n$ of co-dimension $C$ and $\calW$ such that $\calW \oplus \calV^{\perp} = \F_2^n$, there exists a set $S\subseteq \calW$ of size at least $n-C$
such that for every $u\in S$,
$$|\left(u + \calV^{\perp}\right)^{=1}| \geq 1.$$
Moreover, there exists a subset $S_1\subseteq S$ of size at least $n-2C$ whose corresponding shifts contain \emph{exactly} one standard basis vector.
\end{claim}
\begin{proof}

Let $S = \{u: u\in \calW \text{ and } |u+\calV^{\perp}|^{=1} \geq 1\}$. Since every standard basis vector can be expressed as $u+v$ for some $u\in S$ and $v\in \calV^{\perp}$, we have that $\dim(\spn(S\cup \calV^{\perp})) = n$. However, we also know that $\dim(\spn(S\cup \calV^{\perp})) \leq |S| + C$, and rearranging we get
$|S|\geq n- C$.
Next, let $S_1 =\{u\in S: |u+\calV^{\perp}|^{=1} = 1\}$. By \Cref{fact:cosets-unique}, for any $u,u'\in S$, we have $u+\calV^{\perp} \neq u'+\calV^{\perp}$.
Therefore, 
$$n\geq \sum_{u\in S} |(u+\calV^{\perp})^{=1}| \geq \sum_{u\in S_1} |(u+\calV^{\perp})^{=1}| + \sum_{u\in S\setminus S_1} |(u+\calV^{\perp})^{=1}| \geq |S_1| + 2(|S|- |S_1|),$$
and rearranging, we get $|S_1|\geq 2|S| - n \geq n- 2C$. 
\end{proof}

\begin{lemma}
\label{claim:degree-one-bounded-lb}
There is a degree one function $f : \F_2 ^n \to [-1,1]$ for which $\regnum(f,\delta) \geq  n/2$, for all $\delta < 1/n$. 
\end{lemma}

\begin{proof}
The counterexample is given by the function $f(x) = \frac{1}{n} \cdot \sum_i (-1)^{e_i \cdot x}$. 
Let $\calV$ be a subspace of $\F_2^n$ of co-dimension $C$, and suppose we restrict the function to the affine subspace $\calU = \alpha + \calV$. 
By
\Cref{claim:std-basis-vecs-in-small-subspace},  
if $C \leq n/2 - 1$, there exists at least two vectors $\gamma,\gamma' \in \calW$ (where $\calW$ is such that $\calW \oplus \calV^\perp = \F_2^n$) such that  $|(\gamma+ \calV^{\perp})^{=1}| = |(\gamma'+ \calV^{\perp})^{=1}| = 1$. 
Assume without loss of generality that $\gamma\neq 0$. 
Then, by \Cref{fact:Fourier-coeff-of f_U-as-a-sum-of-fourier-coeff-of-f-from-vperp}, we have that
$$|\hat{f_\calU}(\gamma)| = \bigabs{\sum_{\eta \in u + \calV^\perp} \hat{f}(\eta)(-1)^{\inner{\eta, \alpha}}} = \frac{1}{n} > \delta,$$
which follows by observing that exactly one of the summands in the last sum corresponds to a weight one vector and is non-zero. Therefore, $\regnum(f,\delta) \geq n/2$.
\end{proof}

We next show how to generalize \Cref{claim:degree-one-bounded-lb} to degree $d$ bounded functions.

\begin{lemma}
\label{claim:degree-d-bounded-lb}
For $d> 2$ and $\delta < \binom{n}{d}^{-1}$, there exists a degree $d$ function $f: \F_2^n \rightarrow [-1,1]$ for which
$\regnum(f, \delta) \geq n- 2dn^{1/(d-1)}$.
\end{lemma}

\begin{proof}
The counterexample is obtained using a probabilistic argument. We consider the homogeneous degree $d$ polynomial with random signs $f_{\bfz}: \F_2^n \to [-1,1]$, defined as 
$$f_\bfz(x) = \sum_{\gamma: \norm{\gamma}_1 = d} \frac{\bfz_\gamma \cdot (-1)^{\inner{\gamma, x}}}{\binom{n}{d}},$$
where each $\bfz_\gamma \sim \pmone$ is a uniformly random sign.

Let $\calV$ be a subspace of $\F_2^n$ of co-dimension $C$, and suppose we restrict $f_{\bfz}$ to an affine subspace $\calU = \alpha + \calV$. 
By \Cref{claim:std-basis-vecs-in-small-subspace}, we have a $S\subseteq \calW$ (where $\calW$ is such that $\calW \oplus\calV^\perp = \F_2^n$) of size at least $k := n-C$ such that $|(u+\calV^{\perp})^{=1}| \geq 1$ for each $u\in S$. 
Moreover, by \Cref{fact:cosets-unique}, for every $v, v'\in \spn(S)$ we have that $v+\calV^{\perp}\ \neq v' + \calV^\perp$.  
Therefore, there is a set $T\subseteq \calW$ of size at least $\binom{k}{d}$ such that for every $u\in T$, we have $|(u+\calV^\perp)^{=d}| \geq 1$.
By \Cref{fact:Fourier-coeff-of f_U-as-a-sum-of-fourier-coeff-of-f-from-vperp}, for each $u\in T$, we have 
$$|\hat{f_\calU}(u)| = \bigabs{\sum_{\eta\in u + \calV^\perp}\hat{f}(\eta)(-1)^{\inner{\eta,\alpha}}}.$$
We now observe that if $(u+\calV^{\perp})^{=d}$ has odd size, then $\abs{\hat{f_\calU}(u)} \geq \binom{n}{d}^{-1}$. Therefore, if $f_\calU$ was $\delta$-regular, then for each $u\in T$, it must be that the set $(u+\calV^{\perp})^{=d}$ has even size, and, in particular, that $|(u+\calV^{\perp})^{=d}| \geq 2$.  

Let $\calV$ be a subspace such that each non-trivial affine subspace of $\calV^{\perp}$ has an even number of weight $d$ vectors. For a given affine subspace $\calU = \alpha+\calV$ and a random choice of the signs $\bfz_\gamma$'s, the probability that $f_{\calU}$ is $\delta$-regular is therefore at most $2^{-\binom{k}{d}}$. 
Let $\calB$ (for ``$\calB\text{ad}$") be the event that there is an affine subspace $\calU$ where $f_\calU$ is $\delta$-regular.
We can simply union bound over all possible affine subspaces of dimension at least $k$ to bound the probability of $\calB$.
For any $k$, observe that the number of affine subspaces of dimension $k$ is at most $2^{n(k + 1)}$. Thus, we have
$$\Pr[\calB] \leq \sum_{j = k}^{n} 2^{n(j+1)}\cdot 2^{- \binom{j}{d}} \leq \sum_{j = k}^{n} 2^{n(j+1) - \bigparen{\frac{j}{d}}^d}.$$

Note that $h(x) = n(x+1) - \bigparen{\frac{x}{d}}^d$ is concave in $[0, \infty)$; moreover, a quick calculation shows that it is maximized when $x= d\cdot n^{1/(d-1)}$.
Setting $k = 2dn^{\frac{1}{d-1}}$, our desired probability is at most 
\begin{align*}
    \Pr[\calB] &\leq (n-k) \cdot 2^{n(2dn^{\frac{1}{d-1}} + 1) - 2^d\cdot n^{\frac{d}{d-1}}} \tag{Every term is smaller than the first.}\\
    &\leq (n-k) \cdot 2^{(2d + 1 - 2^d)n^{1 + \frac{1}{d-1}}} \\
    &\leq o(1). \tag{$2d+1 - 2^d \leq -1 \quad \forall d \geq 3$.}
\end{align*}

Therefore, there exists a signing $\bfz_\gamma$ such that for any affine subspace of dimension at least $2dn^{1/(d-1)}$, the restriction of $f_{z}$ is not $\delta$-regular.
\end{proof}

\begin{remark}
Note that \Cref{claim:degree-d-bounded-lb} is trivial when $d=2$; it would be interesting to obtain a tighter result in this case.
\end{remark}

\subsection{Boolean Functions}
This section has two parts. The first gives non-explicit lower bounds on $\regnum(f,\delta)$ for Boolean functions, and the second gives explicit lower bounds. 
\subsubsection{Non-explicit Lower Bounds on $\regnum(f, \delta)$}

We can turn our lower bounds on $\regnum(f,\delta)$ for bounded functions into (non-explicit) lower bounds for Boolean functions. To do so, we use the following simple but powerful lemma of \cite{hlms-regularity-lb}, which states that given a bounded function with a large $\regnum(f,\delta)$, there must exist some Boolean function $g$ with similarly a large $\regnum(g,2\delta)$.

\begin{lemma}[\cite{hlms-regularity-lb}, Claim 1.2]
\label{lemma:bounded-regularity-to-boolean-regularity}
Let $\tau > 0$ and $f: \F_2^n \to [-1,1]$. There exists a Boolean function $g : \F_2^n \to \pmone$ satisfying, for every affine subspace $\calU$ such that $|\calU| \geq \frac{4n^2}{\tau^2}$ and any $\gamma \in \F_2^n$, that
$$\bigabs{\hat{f}_\calU(\gamma) - \hat{g}_\calU(\gamma)} \leq \tau.$$
\end{lemma}

\begin{proof}

Let $g(x)$ equal $1$ with probability $\frac{1+f(x)}{2}$, and $-1$ otherwise. Let $\calU = \alpha + \calV$ for some subspace $\calV$. 

By \Cref{fact:coefficients-as-expectation} we can write
$$\hat{f_\calU}(\gamma) = (-1)^{\inner{\gamma, \alpha}}\E_{y \in \calU}[f(y)\cdot (-1)^{\inner{ \gamma, y}}].$$
Consider the random variable 
$$\hat g_\calU(\gamma) = (-1)^{\inner{\gamma, \alpha}}\E_{y \in \calU}[g(y)\cdot (-1)^{\inner{ \gamma, y}}].$$
Observe that $\E_g\hat g_\calU(\gamma)= \hat f_\calU(\gamma)$. 
Moreover, every term in the summation is in $[-1,1]$, so by a Hoeffding bound (see \Cref{fact:hoeffding}), the probability $\left | \hat g_\calU(\gamma ) - \hat f_\calU(\gamma)\right | \geq \tau$ is at most $2\exp\bigparen{-\tau^2|\calU|^2/2} \leq 2^{-2n^2 +1}$.

On the other hand, there are at most $2^{n^2}$ affine subspaces of $\F_2^n$, and at most $2^n$ choices for $\gamma$. Therefore, by a union bound, the probability that $g$ has the property we desire is at least $1 - 2^{n^2 + n - 2n^2 +1} > 0$, and the claim follows.
\end{proof}

Using \Cref{lemma:bounded-regularity-to-boolean-regularity}, we have the following lemma.

\begin{lemma}
\label{cor:nonexplicit-boolean-lower-bounds}

For all $d\geq 3$ and $\delta < \frac{1}{2}\cdot\binom{n}{d}^{-1}$, there exists a Boolean function $f$ with 
$$\regnum(f,\delta) \geq n - \max\left \{2d\cdot n^{1/(d-1)}, \log \bigparen{16n^2/\delta^2}\right \}.$$

\end{lemma}

\begin{proof}
By \Cref{claim:degree-d-bounded-lb}, there exists a bounded $f$ that is not $\delta$-regular in any affine subspace of dimension at least $2dn^{1/(d-1)}$ for all $\delta < \binom{n}{d}^{-1}$. \Cref{lemma:bounded-regularity-to-boolean-regularity} tells us that there exists a Boolean function $g$ whose Fourier coefficients agree up to an additive error $\delta/2$ with the Fourier coefficients of $f$ on all affine subspaces of dimension at least $\log \bigparen{16n^2/\delta^2}$. Therefore, if $f$ is not $\delta$-regular on all of these affine subspaces, then $g$ is also not $\delta/2$-regular on any of these subspaces.
\end{proof}

  We can plug some parameters into \Cref{cor:nonexplicit-boolean-lower-bounds} and achieve the following more parsable corollary.

\begin{corollary}
\label{cor:nonexplicit-boolean-simpler}
For every $3\leq d \leq  \frac{\log n}{\log \log n + 1}  $ and $\delta = \frac{1}{2}\cdot n^{-d}$, there exists a Boolean function $f$ with 
$\regnum(f,\delta) \geq n - 2d\cdot n^{1/(d-1)}$.
\end{corollary}

\begin{proof}
 The function is the same as in \Cref{cor:nonexplicit-boolean-lower-bounds}. We argue that by our choice of parameters, $k$ is always maximized by the first term. We first note that $\frac{1}{2}\cdot \binom{n}{d}^{-1} > \frac{1}{2}\cdot n^{-d} = \delta$, so our choice for $\delta$ is valid. Next, we have that

$$\log (16 n^2/\delta^2) = 5 + 2\log n + 2d\log n \leq  3d \log n \leq 3 \frac{\log^2n}{\log \log n } ,$$
where we used the fact that $d \geq 3$ and $n$ is sufficiently large. On the other hand, note that the function $h(x) = 2xn^{1/(x-1)}$ is decreasing when $x \leq \frac{\log n}{\log \log n + 1}$. Therefore, we have that
\begin{align*}
    2d n ^{1/(d-1)} &\geq 2\cdot \frac{\log n}{\log \log n}\cdot n^{\frac{\log \log n + 1}{\log n}} = 4\cdot \frac{\log^2 n}{\log \log n}. 
\end{align*}

Therefore, the first term is the larger term in \Cref{cor:nonexplicit-boolean-lower-bounds}, as desired.
\end{proof}

\subsubsection{Explicit Lower Bounds on $\regnum(f, \delta)$}
\begin{lemma}[Related to Corollary 1.1 in \cite{parity-kill}]
\label{claim:superlinear-boolean}
For each $\delta > 0$, there exists an explicit Boolean function $f: \F_2^n \to \{0,1\}$ with $\regnum(f,\delta) = \Omega\left((\log \frac{1}{\delta})^{\log_2(3)}\right )$.
\end{lemma}


The proof of \Cref{claim:superlinear-boolean} is based on \Cref{thm:parity-kill-composition}, which appeared in a slightly weaker form in \cite{parity-kill}. 

\begin{theorem}[\cite{parity-kill}]
    \label{thm:parity-kill-composition}
    Let $f: \F_2^n \to \F_2$, and $g: \F_2^m \to \F_2$. We have that 
    $$\paritykill{f \circ g} \geq \paritykill{f} + C_{\min}[f]\cdot B_g,$$
     where $B_g = \max\{\log \paritykill{g} -1, 1\}$.
\end{theorem}
In fact, in \cite{parity-kill} \Cref{thm:parity-kill-composition} appeared as 
$$\paritykill{f \circ g} \geq \paritykill{f} + C_{\min}[f],$$
but they assumed only that $\paritykill{g} \geq 2.$ Therefore, the above result is strictly stronger for any $g$ such that $\paritykill{g}> 4$. We include a proof of this slightly stronger fact in \Cref{appendix:paritykill}.



    
    

We require the following corollary of \Cref{thm:parity-kill-composition}.

\begin{corollary}
\label{cor:supermultiplicativity-of-paritykill}
We have that 
$$\paritykill{f^{\circ k}} \geq B_g \cdot \frac{C_{\min}[f]^k - C_{\min}[f]}{C_{\min}[f] - 1} + \paritykill{f} \geq B_g \cdot C_{\min}[f]^{k-1},$$
where $B_g = \max\{\log \paritykill{g} -1, 1\}$.
\end{corollary}

\begin{proof}[Proof that \Cref{thm:parity-kill-composition} implies \Cref{cor:supermultiplicativity-of-paritykill}]

Let $f = f^{\circ (k-1)}$ and $g=f$. We have by the theorem that

\begin{align*}
    \paritykill{f^{\circ k}} &\geq \paritykill{f^{\circ (k-1)}} + C_{\min}[f^{\circ (k-1)}] \cdot B_g \\
    &\geq \paritykill{f^{\circ (k-1)}} + C_{\min}[f]^{k-1} \cdot B_g \tag{Supermultiplicativity of $C_{\min}$, see \cite{tal-composition}} \\
    &\geq B_g \cdot \sum_{i=1}^{k-1} C_{\min}[f]^{i} + \paritykill{f} \\
    &= B_g \cdot \frac{C_{\min}[f]^k - C_{\min}[f]}{C_{\min}[f] - 1} + \paritykill{f} \geq B_g \cdot C_{\min}[f]^{k-1}. \qedhere
\end{align*}

\end{proof}

For our application, we make the following crucial observation: if $f$ has Fourier coefficients that are all of equal magnitude $\delta$, then any restriction to an affine subspace results in Fourier coefficients of the restricted function that are integer multiples of $\delta$. Hence, if $f$ is $\delta'$-regular, for any $\delta' < \delta$, then $f$ is in fact, constant. In this scenario, finding a subspace in which $f$ is $\delta$-regular is equivalent to finding a subspace where it is constant.

\begin{proof}[Proof of \Cref{claim:superlinear-boolean}]
Consider the following function $g: \F_2^n \to \pmone$:


$$g(x_1, x_2, x_3, x_4) = \frac{1}{2}(-1)^{x_1 + x_3} + \frac{1}{2}(-1)^{x_2 + x_3} + \frac{1}{2}(-1)^{x_1 + x_4} -\frac{1}{2}(-1)^{x_2 + x_4}.$$
Define the function $f:= \frac{1-g}{2}$ so that $f: \F_2^n \to \F_2$.
In other words, f is equal to ${x_1 + x_3}$ if $x_1 = x_2$ and $x_1 + x_4 + \indicator\{x_1 = 0\}$ otherwise. Note that $f$ is a degree 2 function, where all non-zero Fourier coefficients have the same magnitude. We also claim that $\paritykill{f} = 2$. Indeed, we can fix $x_1 + x_2 = 0$ and $x_1 + x_3 = 0$, and we know $f$ equals 0. On the other hand, we have that $C_{\min}[f] = 3$. 

We examine $f^{\circ k} = f(f_1,f_2,f_3, f_4)$, where the $f_i$'s are copies of $f^{\circ (k-1)}$ over disjoint sets of inputs. We claim by induction that $\deg (f^{\circ k}) = 2^k$. This is clearly true when $k = 1$, and for the inductive step we can write
$$f(f_1,f_2,f_3,f_4) = \begin{cases}
f_1 + f_3 & \text{ if }f_1 = f_2 \\
f_4  + f_1 & \text{ if }f_1 \neq f_2.
\end{cases}
$$
Since $\indicator\{f_1 = f_2\} = f_1 + f_2 + 1$, we can write
$$f(f_1,f_2,f_3,f_4) = (f_1 + f_3)(f_1 + f_2 + 1) + (f_1 + f_4)( f_1 + f_2).$$

Therefore, by the inductive hypothesis and the fact that the $f_i$'s are supported over disjoint variables, we have that $\deg f^{\circ k} = 2\cdot \deg f^{\circ (k-1)} = 2^k$. Therefore,\footnote{See \cite{AOBFbook}, Exercise 1.9 or \Cref{claim:granular-coefficients}.} all the Fourier coefficients are integer multiples of $1/2^{2^k}$. So, to make $f$ $\delta$-regular for $\delta < 1/2^{2^k}$, it must be fixed to a constant. Suppose we set $\delta = 1/2^{2^k + 1} $.
By \Cref{cor:supermultiplicativity-of-paritykill}, we have that 
$$\paritykill{f^{\circ k}} \geq C_{\min}[f]^{k-1} = 3^{k-1} = \frac{1}{3}\cdot (2^{k})^{\log_2(3)} = \frac{1}{3}\cdot \frac{1}{2^{\log_2(3)}} \log(1/\delta)^{\log_2(3)} = \frac{1}{9}\cdot \log(1/\delta)^{\log_2(3)}. \qedhere$$

\end{proof}

We now show that the majority function, denoted by $\MAJ_n$, also has a large $\regnum(f,\delta)$ value when $\delta = O(1/\sqrt{n})$.
\begin{lemma}
\label{lemma:majority-lower-bound}
There is an absolute constant $C> 0$, such that for all sufficiently large $n$, $\regnum(\MAJ_n,\delta) \geq \Omega(n^{1/2})$ for any $\delta \leq C/\sqrt{n}$.
\end{lemma}

We need the following three claims to prove this lemma.

\begin{claim}[Fourier Spectrum of $\MAJ_n$, Corollary of Theorem 5.19 in \cite{AOBFbook}]
\label{claim:maj-fourier-coeff}
Consider $f = \MAJ_n$. Each of the following hold. 
\begin{enumerate}
    \item For each $t\in \mathbb{N}$ and $\gamma\in \F_2^n$ with $\|\gamma\|_1 = t$, 
\begin{align*}
\abs{\hat{f}(\gamma)} \leq
\begin{cases}
\left(\frac{t}{n}\right)^{\frac{t-1}{2}}\abs{\hat{f}(e_1)}, &\text{ if $t$ is odd}, \\
0 &\text{ otherwise.}
\end{cases}
\end{align*}
\item For any $\gamma$ with $\|\gamma\|_1 = 1$, $\abs{\hat{f}(\gamma)} \geq \sqrt{\frac{2}{\pi n}}$.
\item For any $\gamma, \gamma'$ such that $\norm{\gamma}_1 + \norm{\gamma'}_1 = n+1$, it holds that $\bigabs{\hat{f}(\gamma)} = \bigabs{\hat{f}(\gamma')}$.
\end{enumerate}
\end{claim}

\begin{claim}[\cite{low-hamming-weight-in-subspaces-math-overflow}]
\label{claim:low-weight-vecs-in-small-subspace}
Let $\calU = \alpha + \calW$ be any affine subspace of $\F_2^n$. For every $t \in [n]$,  let $t^* := \min\{t, n-t\}$. Then, it holds that $$|\calU^{=t}| \leq \binom{\dim(\calW) + 1}{\leq t^*}.$$
\end{claim}

\begin{lemma}\label{lemma:low-weight-vecs-in-small-subspace}
Let $\calV$ be a subspace of $\F_2^n$ of co-dimension $C$ and $\calW$ be such that $\calW \oplus \calV^\perp = \F_2^n$. For each $\ell \leq C+1$, there exists $S_{\ell}\subseteq \calW$ such that $|S_{\ell}| \geq  n - C(\ell+1)$, and for each $\gamma\in S_\ell$ the following two hold:
\begin{enumerate}
    \item $|\left(\gamma + \calV^{\perp}\right)^{=1}| = 1$ and
    \item $|\left(\gamma + \calV^{\perp}\right)^{=t}| \leq 2\cdot \binom{2C+1}{t-1}$, for each $t\leq \ell$.
\end{enumerate}
\end{lemma}

\Cref{claim:maj-fourier-coeff} and \Cref{claim:low-weight-vecs-in-small-subspace} are powerful enough by themselves to achieve a weaker form of \Cref{lemma:majority-lower-bound}: one can use them to show that $\MAJ_n$ is not $\Omega(n^{-1/2})$-regular in any subspace of co-dimension $O(n^{1/3})$.\footnote{Following the proof sketch in \Cref{appendix:lb-sketches}, the reason the analysis breaks if we try to use only \Cref{claim:low-weight-vecs-in-small-subspace} and set $C = n^{1/3 + \eps}$ for any $\eps > 0$ is as follows. By \Cref{claim:low-weight-vecs-in-small-subspace}, there could be on the order of $C^3 = n^{1+3\eps}$ weight three vectors in our signed sum corresponding to the new Fourier coefficient. Since $|\hat{\MAJ_n}(\gamma)| = \Theta(n^{-3/2})$ when $\norm{\gamma}_1 = 3$, these coefficients could combine constructively to a magnitude of $\approx n^{1+3\eps} \cdot n^{-3/2} \gg n^{-1/2}$, thus potentially cancelling out the (single) level one coefficient, which has magnitude $|\hat{\MAJ_n}(e_1)| = \Theta(n^{-1/2})$. }
We now use \Cref{claim:maj-fourier-coeff}, \Cref{claim:low-weight-vecs-in-small-subspace}, and \Cref{lemma:low-weight-vecs-in-small-subspace}, the proofs of which are deferred to \Cref{appendix:omitted}, to prove \Cref{lemma:majority-lower-bound}. 

\begin{proof}[Proof of~\Cref{lemma:majority-lower-bound}]
Let $\calV$ be a subspace of $\F_2^n$ of co-dimension $C = \frac{\sqrt{n}}{10e}-1$, and suppose we restrict $\MAJ_n$ to the affine subspace $\calU = \alpha + \calV$. 
Applying~\Cref{lemma:low-weight-vecs-in-small-subspace} with $\ell = C+1$, we get a subset $S\subseteq \calW$ (where $\calW$ is such that $\calW \oplus\calV^{\perp} = \F_2^n$) of size at least $3$ such that each element $\gamma \in S$ satisfies both items in the lemma. 
In particular, there must be $u\in S$ that satisfies both properties as well as, $\mathbf 0 \not\in u + \calV^{\perp} $ and $\mathbf 1 \notin u + \calV^{\perp}$.\footnote{It is vital that $\mathbf 1 \not \in M\alpha^* + \calV^{\perp}$ since $|\hat{\MAJ_n}(\mathbf{1})| = |\hat{\MAJ_n}(e_i)|$, so they could cancel each other out.} For notational ease, let us denote $E:=u+\calV^{\perp}$. By \Cref{fact:Fourier-coeff-of f_U-as-a-sum-of-fourier-coeff-of-f-from-vperp}, we have
\begin{align}\label{eq:formula-restricted-Fourier-coeff-1}
    \abs{\hat{f_{\calU}}(u)} =
    \bigabs{\sum_{\eta \in E} \hat{f}(\eta)(-1)^{\eta,\alpha}}  
    &\geq \abs{\hat{f}(e_1)} - \sum_{t>1}^{n-1}|E^{=t}|\cdot\abs{\hat{f}\left(\sum_{i=1}^te_i\right)} \nonumber \\ &= \abs{\hat{f}(e_1)} -  \sum_{t>1}^{\frac{n+1}{2}}\left(|E^{=t}| + |E^{=n-t+1}|\right)\cdot\abs{\hat{f}\left(\sum_{i=1}^te_i\right)}.
\end{align}
In the second to last step, we used the facts that majority is a symmetric function and that $|E^{=1}|=1$ and $|E^{=n}|=0$. In the last step, we used item 3 of \Cref{claim:maj-fourier-coeff}. Next, we claim that
\begin{align}
|E^{=t}|+|E^{=n-t+1}| 
\leq \begin{cases}
    (t+1)\binom{2C + 1}{t-1} & \text{ when }t \leq \ell,\\
    2^{C+1} & \text{ otherwise}.
\end{cases}~\label{eq:product-estimate}
\end{align}

For the first case, when $t\leq \ell $, by item 2 of~\Cref{lemma:low-weight-vecs-in-small-subspace}, we have $|E^{=t}| \leq 2\binom{2C + 1}{t-1}$. Furthermore, from \Cref{claim:low-weight-vecs-in-small-subspace}, we have $|E^{=n-t+1}| \leq \binom{C+1}{\leq t-1} \leq \binom{2C + 1}{\leq t-1} \leq (t-1)\cdot \binom{2C+1}{t-1}$ for all $1 <t \leq \frac{\sqrt{n}}{10e}$. When $t> \ell$, we note that both $|E^{=t}|$ and $|E^{=n-t+1}|$ are at most $2^{C}$ since the dimension of $\calV^{\perp}$ is $C$, and this is tighter when $t > C+1$.

Using~\Cref{eq:product-estimate} and~\Cref{claim:maj-fourier-coeff}, we can estimate the sum in~\Cref{eq:formula-restricted-Fourier-coeff-1} as
\begin{align*} 
\abs{\hat{f_{\calU}}(u)} \leq \abs{\hat{f}(e_1)}\bigparen{1 - \underbrace{\sum_{t=3}^{\ell} \binom{2C+1}{t-1}\left(\frac{t}{n}\right)^{\frac{t-1}{2}}(t+1)}_A -  \underbrace{\sum_{t>\ell}^{\frac{n+1}{2}}2^{C+1}\left(\frac{t}{n}\right)^{\frac{t-1}{2}}}_{B}}.
\end{align*}

We complete the argument by showing an upper bound on both the above sums. Starting with $A$, and recalling that $C = \frac{\sqrt{n}}{10e} -1$, we see that
\begin{align*}
    A &\leq \sum_{t=3}^{\frac{\sqrt{n}}{10e}} \left(\frac{\sqrt{n}}{5(t-1)}\right)^{t-1}\left(\frac{t}{n}\right)^{\frac{t-1}{2}}\cdot (t+1) \tag{$\binom{n}{k} \leq \bigparen{\frac{en}{k}}^k.$} \\
    &\leq \sum_{t=3}^{\frac{\sqrt{n}}{10e}} \left(\frac{\sqrt{t}}{5(t-1)}\right)^{t-1}\cdot (t+1) \\
    &\leq \sum_{t=3}^{\frac{\sqrt{n}}{10e}} \left(\frac{(t+1)^{\frac{1}{2}}}{2(t+1)}\right)^{t-1} \cdot (t+1) \tag{$5(t-1) \geq 2(t+1) \quad \forall t \geq 3$.} \\
    &= \sum_{t=3}^{\frac{\sqrt{n}}{10e}} \left(\frac{1}{2(t+1)^{1/2}}\right)^{t-1} \cdot (t+1) \leq \sum_{i=2}^{\infty}\bigparen{\frac{1}{2}}^i \leq 1/2.
\end{align*}
In the penultimate inequality, we used the fact that for the first term, when $t=3$, we have $\left(\frac{1}{2(t+1)^{1/2}}\right)^{t-1} \cdot (t+1) = \bigparen{\frac{1}{4}}^2 \cdot 4 = 1/4$, and the ratio of the summands (for $t \geq 3$) is 
$$\frac{(2(t+1)^{1/2})^{t-1}}{(2(t+2)^{1/2})^t}\cdot \frac{t+2}{t+1} \leq \frac{1}{2(t+1)^{1/2}} \cdot 2 \leq \frac{1}{(t+1)^{1/2}} \leq \frac{1}{2}.$$
To bound $B$, we note that the function $h(x) = \left(\frac{x}{n}\right)^{(x-1)/2}$ is strictly convex, which means its maximum occurs either at $t = C$ or $t = \frac{n+1}{2}$. Again, setting $C= \frac{\sqrt{n}}{10e} - 1$, a quick calculation shows that the maximum is achieved for the first term, and this term is at most 
\begin{align*}
     2^{\frac{\sqrt{n}}{10e}} \cdot \bigparen{\frac{1}{10e\sqrt{n}}}^{\frac{\sqrt{n}}{20e}} &\leq 2^{\frac{\sqrt{n}}{10e}} \cdot 2^{-\log n \cdot \frac{\sqrt{n}}{40e}} \\
    &\leq 2^{\frac{\sqrt{n}}{10e} - \bigparen{\frac{\sqrt{n}}{10e} + 2\log n}} \tag{$\frac{\sqrt{n}}{40e}\log n \geq \frac{\sqrt{n}}{10e} + 2\log n$ for large enough $n$.} \\
    &= 2^{-2\log n} = \frac{1}{n^{2}}.
\end{align*}

This implies that $B \leq n \cdot \frac{1}{n^{2}} \leq o(1)$. Using item 2 of \Cref{claim:maj-fourier-coeff}, we conclude that there is a non-trivial Fourier coefficient $$\abs{\hat{f_{\calU}}(u)} \geq \abs{\hat{f}(e_1)}\left(1 -1/2 -  o(1)\right) = \Omega(n^{-1/2}). \qedhere$$
\end{proof}

\section{Applications}
\label{sec:app}
We now present an application of \Cref{thm:bounded-deg-d-upper-bound-intro} that shows a tradeoff between the dimension of a disperser and its Fourier degree, and a connection to extractors, as well. First, we introduce a definition that generalizes Boolean functions and helps us reason about the Fourier spectrum of dispersers. 

\begin{definition}
\label{def:granular-fn}
We say a function $f:\F_2^n \rightarrow \R$ is $G$-granular if for every $x\in \F_2^n$, we have that $f(x)$ is an integer multiple of $G$.
\end{definition}

\begin{claim}
\label{claim:granular-coefficients}
If a degree $d$ function $f: \F_2^n \rightarrow \R$ is $G$-granular, then for every $\gamma \in \F_2^n$, we have that $\hat{f}(\gamma)$ is an integer multiple of $2^{-d}\cdot G$. 
\end{claim}
\begin{proof}

Note that if we associate $\F_2$ with $\{0,1\}$, any $f: \F_2^n \to \R$ has a real multilinear polynomial representation $q: \{0, 1\}^n \to \R$, where $q(x) = f(x)$ for all $x\in \{0,1\}^n$ (see Exercise 1.9 in \cite{AOBFbook}).
In particular, we can write $q$ as a sum of its indicators:
$$q(x) = \sum_{a \in \{0,1\}^n}\indicator\{x= a\}\cdot q(a).$$
Noting that $\indicator\{x=a\} = \prod_i (1 - a_i - x_i)(1-2a_i)$, we see that every coefficient of $q$ is an integer multiple of $G$. 

However, we can also associate $f$ with a real multilinear polynomial, $p :\pmone^n \to \R$, such that $f(x) = p((-1)^x)$ for all $x \in \F_2^n$. Note then, that $p(x) = q((1-x_1)/2, \ldots, (1-x_n)/2))$, so if $p$ has degree $d$, then all its coefficients are integer multiples of $G \cdot 2^{-d}$. Finally, note that $f$ and $p$ have the same Fourier coefficients (and therefore degree), which implies the result.
\end{proof}

We now show that low degree granular functions cannot have a large parity kill number. As a consequence, we get that low-degree affine dispersers cannot have small dimension (\Cref{thm:rule-out-disperser}). 

\begin{lemma}
Every degree $d$ function $f:\F_2^n \rightarrow [-1,1]$ that is $G$-granular satisfies $$\paritykill{f} \leq n - \Omega\left(n^{1/d!}(d + \log n/G)^{-2}\right).$$
\end{lemma}
\begin{proof}
If $f$ is $G$-granular and degree $d$, then from~\Cref{claim:granular-coefficients} we know that all its Fourier coefficients must be integer multiples of $2^{-d}\cdot G$. 
Moreover, a Fourier coefficient of $f$ in any affine subspace is simply a signed sum of the Fourier coefficients of $f$ and therefore it must also be an integer multiple of $2^{-d}\cdot G$. 
This shows that if $f$ is $\delta$-regular in some affine subspace $\calU$ with $\delta < 2^{-d}\cdot G$, then $f_\calU$ must be constant. 
The lemma follows by using \Cref{thm:bounded-deg-d-upper-bound-intro} for $\delta = 2^{-d-1}\cdot G$.
\end{proof}


\begin{proof}[Proof of \Cref{thm:rule-out-disperser}]
Using $f$, we can construct a degree $d$ function $h:\F_2^n\to [-1,1]$ as $h(x) = 1 - \frac{2f(x)}{C}$. Noting that $h$ is $2/C$-granular and using the above lemma, it follows that 
$$\paritykill{f} = \paritykill{h}\leq n - \Omega\left(n^{1/d!}(d+\log (nC))^{-2}\right),$$
which shows that there is some affine subspace of dimension at least $\Omega\left(n^{1/d!}(2d+\log (nC))^{-2}\right)$ where $f$ is constant.
\end{proof}

Last, we give a connection between the notion of $\delta$-regularity and affine extractors. Formally, we define affine extractors as follows.

\begin{definition}[Affine Extractor]
\label{dfn:affine-extractor}
A function $f:\F_2^n \to \{0,\ldots,C\}$ is said to be a $(k, \delta)$-\emph{affine extractor} if for all affine subspaces $\calU$ of dimension at least $k$, we have that $$|f_\calU - \uniform_C| \leq \delta,$$
where $\uniform_C$ is the uniform distribution over $\{0,\ldots,C\}$.
\end{definition}

\begin{claim}
\label{claim:extractor-implies-regular}
If $f$ is a $(k,\delta)$-extractor, then $f$ becomes $2C\delta$-regular when restricted to any affine subspace of dimension at least $k+1$.
\end{claim}
\begin{proof} 
Note that if $f:\F_2^n\to\{0,\ldots,C\}$ is a $(k,\delta)$-extractor then in any affine subspace $\calU$, of dimension at least $k$, we have, $$\bigabs{\hat{f_\calU}(\chi_{\mathbf{0}})-\frac{C}{2}} = \bigabs{\sum_{c}c\left(\Pr_{x\in \calU}[f(x) = c] - \frac{1}{C+1}\right)} \leq C\sum_{c}\bigabs{\Pr_{x\in \calU}[f(x) = c] - \frac{1}{C+1}}\leq 2C\delta.$$
Suppose $f$ is a $(k, \delta)$-affine extractor. Let us assume to a contradiction that $\calU$ is an affine subspace of dimension at least $k+1$, where $f_\calU$ has a Fourier coefficient with magnitude larger than $2C\delta$. By \Cref{cor:V-dim-n-1}, we can fix the parity corresponding to this Fourier coefficient in such a way that the bias of the function increases by $2C\delta$, which gives the desired contradiction. 
\end{proof}

\section{Future Directions}

We highlight two open problems that offer particularly interesting research directions. First, there is a tantalizing, and large, gap between our \Cref{thm:bounded-deg-d-upper-bound-intro} and \Cref{claim:degree-d-bounded-lb} for bounded degree $d$ functions. We suspect that \Cref{claim:degree-d-bounded-lb} is closer to being tight and ask the following question.

\begin{direction}
Can the upper bound on $\regnum(f,\delta)$ in \Cref{thm:bounded-deg-d-upper-bound-intro} be improved? 
\end{direction}

Moreover, it would be interesting to find explicit Boolean and bounded functions with large $r(f,\delta)$ values. 

\begin{direction}
Find (explicit) examples of functions $f:\F_2^n\to [-1,1]$ with $r(f,\delta)$ values comparable to those obtained in \Cref{claim:degree-d-bounded-lb}. Similarly, find (explicit) Boolean functions with similar $r(f,\delta)$ values.
\end{direction}

\section{Acknowledgements}

We thank Anup Rao for posing the question that launched this project and for his invaluable advice and feedback. We are also grateful Paul Beame for his extremely helpful advice, discussions, and feedback. Finally, we thank Sandy Kaplan for detailed feedback on this writeup.

\newpage

\printbibliography

\newpage 

\appendix 

\section{Omitted Sketches}
\label{appendix:lb-sketches}

We give the main ideas behind the lower bounds in \Cref{table:large-reg-nums}. 

\subparagraph*{Sketch of \Cref{claim:degree-one-bounded-lb}.} The proof of this claim is based on the homogeneous degree-one function $f(x) = \frac{1}{n}\sum_i(-1)^{x_i}$. Its key idea comes from \Cref{claim:std-basis-vecs-in-small-subspace}, which we use to show that if the dimension of $\codim(\calV) < n/2$, then at least one shift of $\calV^{\perp}$ must contain \textit{exactly} one standard basis vector. 
By the preceding discussion, this implies that $f_{\alpha+\calV}$ has a non-trivial Fourier coefficient with magnitude exactly $1/n > \delta$.

We remark that \Cref{claim:degree-one-bounded-lb} is tight. The function $f$ is symmetric, and for any such function, we can fix $n/2$ parities to obtain an affine subspace where every vector has weight $n/2$, which in turn fixes the function.
%
\subparagraph*{Sketch of \Cref{claim:degree-d-bounded-lb} and \Cref{cor:nonexplicit-boolean-simpler}.} To achieve \Cref{claim:degree-d-bounded-lb}, one might expect to extend the above argument to the homogeneous degree $d$ function $f(x) = \binom{n}{d}^{-1}\sum_{\gamma: \norm{\gamma}_1 = d}(-1)^{\inner{\gamma, x}}$.
Unfortunately, this function is symmetric, and we have $r(f,0) \leq n/2$. We therefore consider a random homogeneous degree $d$ function 
$f_\bfz(x) = \binom{n}{d}^{-1}\sum_{\gamma: \norm{\gamma}_1 = d}\bfz_\gamma\cdot (-1)^{\inner{\gamma, x}},$
where each $\bfz_\gamma$ is a random sign. 
A simple argument, again utilizing \Cref{claim:std-basis-vecs-in-small-subspace}, shows that there must be at least $\binom{k}{d}$ affine subspaces of $\calV^\perp$ with at least one vector of weight $d$. By our earlier reasoning, each of those subspaces must in fact contain at least two vectors of weight $d$ so that the restricted function would have a non-trivial Fourier coefficient with magnitude $\binom{n}{d}^{-1} >\delta$. Moreover, the probability (over the signs $\bfz_\gamma$'s) that each of the $\binom{k}{d}$ signed sums cancels is at most $2^{-\binom{k}{d}}$, and a union bound over all the possible affine subspaces of dimension $k = \Theta(dn^{1/(d-1)})$ completes the argument.
%

If we restrict our attention to Boolean functions, we might hope to obtain strong upper bounds for $\regnum(f,\delta)$; however, \Cref{cor:nonexplicit-boolean-simpler} rules this out.
The proof of this claim is based on a simple lemma of \cite{hlms-regularity-lb} (\Cref{lemma:bounded-regularity-to-boolean-regularity}), which uses the probabilistic method to convert a bounded function that is not $\delta$-regular in large affine subspaces to a Boolean function with the same property. Applying this lemma to the lower bound from \Cref{claim:degree-d-bounded-lb} achieves the result. 
\subparagraph*{Sketch of \Cref{lemma:majority-lower-bound}.}
This lower bound is based on the majority function. 
Its key idea is that there exists a non-trivial affine subspace of $\calV^{\perp}$ containing \emph{exactly} one weight-$1$ vector and relatively few vectors of higher weight (see \Cref{lemma:low-weight-vecs-in-small-subspace}). 
Then, we use properties of the Fourier spectrum of the majority function to show that the signed sum of the Fourier coefficients of majority corresponding to vectors in this affine subspace, is on the order of $\abs{\hat{f}(e_1)} = \Omega(n^{-1/2})$. 
Specifically, we argue that even if the coefficients coming from higher weight vectors in the aforementioned sum combined in the most constructive way possible, they cannot combine to more than $\bigabs{\hat{f}(e_1)}/2$. 
We also note that \Cref{lemma:majority-lower-bound} is tight up to constant factors via \Cref{claim:trivial-algorithm}.
Conversely, \Cref{lemma:majority-lower-bound} implies that for $\delta\geq n^{-1/2}$, the majority function on $O(1/\delta^2)$ variables is an explicit Boolean function for which $\regnum(f, \delta) \geq \Omega(1/\delta)$. 

\subparagraph*{Rationale for \Cref{claim:superlinear-boolean}.} The last entry in the table corresponds to \Cref{claim:superlinear-boolean} and is based on a simple function $f$ on 4 inputs that is composed with itself $k$ times. We use key properties of the composition of Boolean functions (from \cite{tal-composition, parity-kill}) to achieve the bound. The function itself is the same one considered in \cite{parity-kill}, and we use their main theorem crucially to obtain our lower bound. We present a slightly generalized version of the main theorem of \cite{parity-kill}, so we include a proof in \Cref{appendix:paritykill}.

We make some final comments about the lower bounds from \Cref{cor:nonexplicit-boolean-simpler}. 
The Boolean functions that achieve the lower bounds share the property that the magnitudes of their Fourier coefficients are extremely close to their bounded counterparts in \Cref{claim:degree-d-bounded-lb}. However, even though the bounded functions themselves have low degree, the Boolean functions are very far from being low-degree functions; in fact, \emph{almost all} their Fourier mass comes from the high-degree terms.
Notably, these functions are also non-explicit affine dispersers with small dimension, and it would be interesting to find explicit Boolean functions with similar strong lower bounds on the $\regnum(f,\delta)$. 

\section{Short Proof of the Parity Kill Number Theorem (\cite{parity-kill})}
\label{appendix:paritykill}

We present a more concise and slightly improved version of the main theorem of \cite{parity-kill}, which appears as \Cref{thm:parity-kill-composition} above.

The following proposition suffices to prove the theorem.

\begin{proposition}
\label{prop:main-pkc-prop}
Let $f' : \F_2^n \times \F_2 \to \F_2$ and $g: \F_2^k \to \F_2$. We let $f: \F_2^n \times \F_2^k \to F_2$ be defined as 
$$f(x,y) = f'(x, g(y)).$$
Then for any affine subspace $H \subseteq \F_2^n \times \F_2^k$ on which $f$ is constant, there exists some $H'\subseteq \F_2^n \times \F_2$ on which $f'$ is constant such that either:
\begin{enumerate}
    \item $\codim(H') \leq \codim(H) - B_g$, where $B_g = \max\{1, \log \paritykill{g} - 1\}$, as before.
    \item The $(n+1)$-st coordinate (so $g(y)$) is irrelevant in $H'$ and $\codim(H') \leq \codim(H)$.
\end{enumerate}
Furthermore, among the first $n$ coordinates, any coordinate that was irrelevant in $H$ remains irrelevant in $H'$.
\end{proposition}

Before proving \Cref{prop:main-pkc-prop}, let's see how it implies \Cref{thm:parity-kill-composition}. Note that $f \circ g  = f(g(x_1), ..., g(x_n))$, so we will apply \Cref{prop:main-pkc-prop} $n$ times. The crucial observation is that we must fall into the first case of \Cref{prop:main-pkc-prop} at least $C_{\min}[f]$ times. This is because if $f(x,y)$ is constant on $H$, then $H$ must depend on at least $C_{\min}[f]$ coordinates. 

Let then $H \subseteq \F_2^{n\cdot m}$ be a minimum co-dimension subspace on which $f \circ g$ is constant, so that $\codim(H) = \paritykill{f \circ g}$. Applying \Cref{prop:main-pkc-prop} $n$ times, we derive $H'$ on which $f$ is constant.
\begin{align*}
    \paritykill{f} &\leq \codim(H') \\
    &\leq \codim(H) - B_g \cdot C_{\min}[f] \\
    &= \paritykill{f \circ g} - B_g \cdot C_{\min}[f].
\end{align*}
Rearranging gives the theorem.

Finally, before we prove \Cref{prop:main-pkc-prop}, we need the following lemma, the proof of which is not complicated but we will omit and can be found in \cite{parity-kill}.

\begin{lemma}[\cite{parity-kill}, Lemma 3.3]
\label{lemma:canonization}
Let $H \subseteq \F_2^k \times \F_2^n$ be an affine subspace. Then there exists an invertible linear transformation $L$ on $\F_2^k \times \F_2^n$ such that, after applying this linear transformation, the constraints of $H$ can be partitioned into 
\begin{itemize}
    \item $\calB_{x,y}$, which contain constraints of the form $x_i + y_i = \sigma_i$, for $1 \leq i \leq t$.
    \item $\calB_x$, which contain constraints of the form $x_j = \sigma_j$, for $t+1 \leq j \leq t'$.
    \item $\calB_y$, which contain constraints of the form $y_k = \sigma_k$, for $t'+1 \leq k \leq t''$.
\end{itemize}
and $t + (t'-t) + (t''- t') = \codim(H)$.
\end{lemma}

The takeaway of the above lemma is that since parity kill number is invariant under affine transformations, we can ``canonize" any affine subspace in a way that minimizes the interactions between coordinates.
\begin{proof}[Proof of \Cref{prop:main-pkc-prop}]
WLOG suppose that $H$ is of the form given in \Cref{lemma:canonization}.
\begin{enumerate}
    \item \textbf{Easy Case: $|\calB_{x,y}| = 0$.}

        Let's denote $C_y$ as the set of all $y$ that satisfy the constraints in $\calB_y$, and let $C_x$ (analogously) be the set of $x$ that satisfy the constraints of $\calB_x$.
        \begin{enumerate}[a)]
            \item \textbf{Subcase 1:} Suppose that $g(y) = b$ for all $y \in C_y$. Then we can let 
            $$H' = \bigcurly{(x,z) \bigvert x \in C_x, z = b}.$$
            $f'$ is clearly constant on $H'$. Note that 
            $$\codim(H') = |\calB_x| + 1 = \codim(H) - |\calB_y| + 1 \leq \codim(H) - \paritykill{g} + 1 \leq \codim(H) - B_g,$$
            as desired for the first case of \Cref{prop:main-pkc-prop}.
            
            \item \textbf{Subcase 2:} Suppose that $g$ is not constant on the inputs in $C_y$. In this case, we claim that 
            $$H' = \bigcurly{ (x,z) \bigvert x \in C_x}$$
            makes $f'$ constant. Indeed, suppose it doesn't. Then there are two inputs $(x, z)$ and $(x',z')$ such that $f'(x, z) \neq f'(x', z')$. But then, we can pick $y, y' \in C_y$ such that $g(y) = z$ and $g(y') = z'$, and this results in $(x,y), (x',y') \in H$ such that $f(x,y) \neq f(x',y')$, a contradiction. 
            
            Finally, note that $\codim(H') = |\calB_x| = \codim(H) - |\calB_y| \leq \codim(H) - B_g$. In fact, we don't even need this to be true in order to fall into the second case of the proposition (since $H'$ does not depend on its last coordinate), but it is nonetheless true.

    \end{enumerate}
    
    \item \textbf{(Slightly) Harder Case: $|\calB_{x,y}| \neq 0$.}
    
        \begin{enumerate}[a)]
            \item \textbf{Subcase 1:} $g$ becomes a junta on $y_1,...,y_t$ when restricted to $C_y$. In this case, let $I = \{ i_1, ..., i_s\} \subseteq [t]$ be the junta variables, so that $g(y) = h(y_{i_1},...,y_{i_s})$ for all $y \in C_y$. Then we claim that $f'$ is constant on 
            $$H' = \bigcurly{(x,z) | x \in C_x, x_i = 0 \: \forall i \in [t] \setminus I,  z = h(x_{i_1} \oplus \sigma_{i_1}, ... x_{i_s} \oplus \sigma_{i_s})}.$$
            Indeed, suppose it is not, so that $f'(x, z) \neq f'(x', z')$. Take $y \in C_y$ such that $y_i = \sigma_i \: \forall i \in [t] \setminus I$, and $y_{j} = \sigma_j \oplus x_j \: \forall k \in I$. Then we have that $g(y) = h(y_{i_1}, ... ,y_{i_s}) = z$. Similarly, we can find $y' \in C_y$ such that $g(y') = z'$. We end up at a contradiction though, since $(x,y)$ and $(x', y')$ are both in $H$, but are such that $f(x,y) \neq f(x',y')$.
            
            Finally, note that the codimension of $H'$ is exactly $|\calB_x| + |\calB_{x,y}| - s + 1 = \codim(H) - (|\calB_y| + s) + 1$. Next, we claim that $|\calB| + s \geq \log \paritykill{g}$. To see why this is the case, note that we can fix $g$ by fixing at most $|\calB_y| + 2^s$ parities/variables. This implies that $|\calB_y| + 2^s \geq \paritykill{g}$ which implies that $|\calB_y| + s \geq \log \paritykill{g}$. Thus, we have that
            $$\codim(H') \leq \codim(H) - B_g,$$
            as desired.
            
            \item \textbf{Subcase 2:} There exists some $b_1,...,b_t$ such that $g(y)$ is \emph{not} constant on 
            $$C_y' := \bigcurly{y \left | \substack{y \in C_y \\ y_i = b_i \: \: \forall \: 1 \leq i \leq t} \right .}.$$
            
            In this case, let 
            $$H' = \bigcurly{(x,z)\:  |\:  x \in C_x, \: \: x_j = b_j \oplus \sigma_j \:\:  \forall \: 1 \leq j \leq t}.$$
            
            First, we claim that $f'$ is constant on $H'$. As before, suppose it is not, so that $f'(x, z) \neq f'(x', z')$. Then by definition, there exists $y,y'$ such that $y_i = y_i' = b_i$ for all $i \in [t]$, such that $g(y) = z$ and $g(y') = z'$. In this case, $(x,y)$ and $(x', y')$ are both in $H$, but are such that $f(x,y) \neq f(x', y')$, a contradiction. 
            
            Finally, note that $\codim(H') \leq \codim(H)$, but that $H'$ is independent of its last coordinate $z$, so that we fall into the second case of \Cref{prop:main-pkc-prop}.
        \end{enumerate}
\end{enumerate}
\end{proof}

\section{Omitted Proofs}
\label{appendix:omitted}


\subsection{Proofs of \Cref{claim:ghz-partition} and \Cref{claim:trivial-algorithm}}

In this section we provide the proofs of \Cref{claim:ghz-partition} and \Cref{claim:trivial-algorithm}. We first begin with a corollary of \Cref{fact:Fourier-coeff-of f_U-as-a-sum-of-fourier-coeff-of-f-from-vperp} which will be useful in the analysis of the claims.    
    
\begin{corollary}
\label{cor:V-dim-n-1}
When $\calV$ has dimension $n-1$, this corresponds to fixing a single parity $\sum_{i: \gamma_i = 1}x_i$ to $b \in \{0,1\}$. Then $\calV^\perp$ is simply $\spn(\{\gamma\})$ and $\alpha$ is any vector such that $\inner{\gamma, \alpha} = b$. 

Then for all $\gamma' \neq \gamma$ we have by \Cref{fact:Fourier-coeff-of f_U-as-a-sum-of-fourier-coeff-of-f-from-vperp} that 
$$\hat{f_{\alpha + \calV}}(\chi_{\gamma'}) = (-1)^{\inner{\gamma', \alpha}}\hat{f}(\gamma') + (-1)^{\inner{\gamma + \gamma', \alpha}} = (-1)^{\inner{\gamma', \alpha}}\left (\hat{f}(\gamma') + (-1)^b\cdot \hat{f}(\gamma + \gamma')\right ).$$
In particular, there exists a choice of $b$ such that
$$\left |\hat{f_{\alpha + \calV}}(\chi_\mathbf{0})\right | =\left | \hat{f}(\mathbf{0}) + \hat{f}(\gamma)\right | .$$
\end{corollary}    
    
\begin{proof}[Proof of \Cref{claim:ghz-partition}]
    Given some $f: \F_2^n \to [-1,1]$, consider the following simple procedure:
    \begin{itemize}
        \item While at least $\delta$ fraction of $\pi \in \Pi$ have some $\gamma_\pi$ such that $|\hat{f_\pi}(\gamma_\pi)| > \delta$, further partition each $\pi$ into $\pi \cap \{x : \inner{\gamma_\pi, x} = 0\}$ and $\pi \cap \{x : \inner{\gamma_\pi, x} = 1\}$. 
    \end{itemize}
    We would like to show that we cannot perform the above partitioning action more that $\frac{1}{\delta^3}$ times. 
    Towards this end, define the potential function $\Phi(\Pi) := \E_{\pi \in \Pi} \hat{f}_\pi(0)^2 = \E_{\pi\in\Pi}[(\E f_\pi)^2] \in [0,1]$. 
    Whenever we partition further, by \Cref{cor:V-dim-n-1} each $|\hat{f}_\pi(0)|$ is updated to either $|\hat{f}_\pi(0) + \hat{f}_\pi(\gamma_\pi)|$ or $|\hat{f}_\pi(0) - \hat{f}_\pi(\gamma_\pi)|$.
     Therefore, the contribution of $\pi$ to $\Phi$ in one step of the partitioning process is
     $$\frac{1}{2}\bigparen{(\hat{f}_\pi(0) + \hat{f}_\pi(\gamma_\pi))^2 + (\hat{f}_\pi(0) - \hat{f}_\pi(\gamma_\pi))^2} - \hat{f}_\pi(0)^2 = \hat{f}_\pi(\gamma_\pi)^2.$$
     
     Since we assume at least $\delta$ fraction of $\pi \in \Pi$ had some $\gamma_\pi$ such that $|\hat{f_\pi}(\gamma_\pi)| > \delta$, at each step of the refinement $\Phi$ must increase by at least $\delta^3$, completing the proof.
\end{proof}

\begin{proof}[Proof of \Cref{claim:trivial-algorithm}]
    Suppose without loss of generality, $\E f \geq 0$. Start with the trivial subspace, $\pi_0 = \F_2^n$. While there exists $\gamma$ such that $|\hat f_{\pi_t}(\gamma)| > \delta$, by \Cref{cor:V-dim-n-1} we can fix the parity corresponding to $\gamma$ in such a way that ensures that $|\hat f_{\pi_{t+1}}(0)| = \left |\hat{f}_{\pi_t}(0) + |\hat f_{\pi_t}(\gamma)|\right | > \hat{f_{\pi_t}}(\gamma) + \delta$. Since $\hat{f}_\pi(0) \leq 1$ for all $\pi$, this process can happen at most $\frac{1}{\delta}$ times.
\end{proof}

\subsection{\Cref{claim:maj-fourier-coeff}}

\begin{proof}[Proof of \Cref{claim:maj-fourier-coeff}]
Theorem 5.19 in \cite{AOBFbook} gives the following formula for the Fourier coefficients of the Majority function:
$$\bigabs{\hat{\MAJ_n}(\gamma)} =  \frac{\binom{\frac{n-1}{2}}{\frac{t-1}{2}}}{\binom{n-1}{t-1}}\cdot \frac{2}{2^n}\binom{n-1}{\frac{n-1}{2}},$$

which holds for all $\gamma$ such that $\norm{\gamma}_1 = t$ is odd. Otherwise, $\hat{\MAJ_n}(\gamma) = 0$.
By the above equation, we have that

\begin{align*}
    \frac{\hat{\MAJ_n}(\gamma)}{\hat{\MAJ_n}(e_1)} &= \frac{\binom{\frac{n-1}{2}}{\frac{t-1}{2}}}{\binom{n-1}{t-1}} \\
    &= \frac{\bigparen{\frac{n-1}{2}}! \cdot (t-1)!\cdot (n-t)!}{\bigparen{\frac{t-1}{2}}!\cdot \bigparen{\frac{n-t}{2}}! \cdot(n-1)!}\\
    &=\frac{(t-2)!! \cdot (n-t-1)!!}{(n-2)!!} \\
    &= \frac{(t-2)\cdot (t-4) \cdots 1}{(n-2) \cdot (n-4) \cdots (n-t+1)} \\
    &\leq \bigparen{\frac{t}{n}}^{\frac{t-1}{2}}.  \qedhere
\end{align*}
\end{proof}

\subsection{\Cref{claim:low-weight-vecs-in-small-subspace}}

\begin{proof}[Proof of \Cref{claim:low-weight-vecs-in-small-subspace}]
First, consider
$$\Tilde{\calU} = \begin{cases}
\calU & \text{ if } t\leq n/2 \\
\mathbf{1} + \calU & \text{ if }t>n/2.
\end{cases}$$
Note that $\dim(\Tilde{\calU})\leq \dim(\calV) + 1$. Moreover, note that $|\calU^{=t}| = |\Tilde{\calU}^{=t^*}|$.
Using Gaussian elimination, we can find a basis find a basis $b_1,\ldots,b_k$ for $\calU$ such that $N(b_1)< N(b_2)<\ldots<N(b_k)$, where $k = \dim(\calU)$ and $N(b) := \min_i \{ i : b_i \neq 0\}$. Moreover (again via Gaussian elimination), we can ensure that $b_i$ is the only basis vector with a $1$ in entry $N(b_i)$. Therefore, any vector in $\calU$ involving more than $t^*$ basis vectors must have more than $t^*$ nonzero entries. Therefore, we have that
$$|\calU^{=t}| \leq |\calU^{\leq t}| = |\Tilde{\calU}^{\leq t^*}| \leq \sum_i^{t^*}\binom{\dim(\Tilde{\calU})}{i} \leq \binom{\dim(\calV) + 1}{\leq t^*}. \qedhere $$
\end{proof}

\subsection{\Cref{lemma:low-weight-vecs-in-small-subspace}}

\begin{proof}[Proof of~\Cref{lemma:low-weight-vecs-in-small-subspace}]
We will prove the statement by induction on $\ell$. Setting $S_1$ to be $\{\gamma_1,\ldots,\gamma_{n-2C}\}$ guaranteed by \Cref{claim:std-basis-vecs-in-small-subspace} such that $|(\gamma + \calV^\perp)^{=1}|=1$ for all $\gamma \in S_1$ proves the base case when $\ell = 1$. 

Now suppose we have some $S_{\ell-1}$ that satisfies the conditions in the lemma. 
We will pick $S_\ell \subseteq S_{\ell-1}$ that satisfies condition (2) for $t= \ell$, and argue that the number that do not satisfy the condition is at most $C$. 
Indeed, suppose towards a contradiction that $|S_{\ell-1}\setminus S_\ell|\geq C+1$. Let $J\subseteq S_{\ell-1}\setminus S_\ell$ be any subset of size $C+1$ and $H := \bigcup_{\gamma \in J}(\gamma + \calV^\perp)^{=\ell} $. 
Since $S_\ell \subseteq \calW$, we can say by \Cref{fact:cosets-unique} that the sets $(\gamma+\calV^{\perp})_{\gamma\in S_{\ell}}$ are all mutually disjoint and therefore, 
$$|H| = \sum_{\gamma \in J} |(\gamma + \calV^\perp)^{=\ell}| >  2(C+1)\binom{2C+1}{\ell-1}= \ell\binom{2C+2}{\ell}.$$

However, $H \subseteq (\spn(\{\gamma: \gamma \in J\}\cup V^{\perp} ))^{=\ell}$. Since, $\dim(\spn(\{\gamma: \gamma \in J\}\cup V^{\perp} )) \leq |J| + C = 2C+1$, by 
\Cref{claim:low-weight-vecs-in-small-subspace} it must be that $|H| \leq \binom{2C+2}{\leq \ell} \leq \ell \binom{2C+2}{\ell}$, where the inequality holds for all $\ell \leq (2C+2)/2 = C+1$. This is a contradiction, and we conclude that $|S_{\ell-1}\setminus S_\ell| \leq C$.
\end{proof}

\end{document}